\documentclass[10pt]{IEEEtran}
\ifCLASSINFOpdf
\else
\fi
\usepackage{graphicx}
\usepackage[centertags]{amsmath}
\usepackage{amsmath}
\usepackage{bbold}
\usepackage{float}
\usepackage{amsthm}
\newtheorem*{remark}{Remark}
\usepackage{algorithmic}

\usepackage{textcomp}
\newtheorem{definition}{Definition}
\usepackage{verbatim}
\usepackage{color}
\usepackage{graphicx}
\usepackage{tabularx}
\newcolumntype{P}[1]{>{\centering\arraybackslash}p{#1}}
\newcommand\scalemath[2]{\scalebox{#1}{\mbox{\ensuremath{\displaystyle #2}}}}
\usepackage{cite}
\usepackage{multicol}
\usepackage{amsmath,amssymb,amsfonts}
\usepackage{amsthm} 
\usepackage{tikz}
\usetikzlibrary{calc}

\usepackage{xcolor}
\newtheorem{lemma}{Lemma}
\usepackage{tabularx,booktabs}
\usepackage{dblfloatfix}
\usepackage{array}
\newcolumntype{C}{>{\centering\arraybackslash}X} 
\usepackage{lipsum}
\begin{document}
		%
		\title{AULAs: A Novel Family of Augmented ULAs for Enhanced Localization of Non-Circular Sources with Reduced Mutual Coupling Effects\\} 
		\author {Abdul Hayee Shaikh, and  
			Xiaoguang Liu$^\ast$, \textit{Senior Member}, \textit{IEEE}} 
	\maketitle   
	\footnotetext[1]{	The authors are with the School of Microelectronics, Southern University of Science and Technology, Shenzhen, China.
	}
	\markboth{Journal of \LaTeX\ Class Files,~Vol.~14, No.~8, August~2021}%
	{Shell \MakeLowercase{\textit{et al.}}: A Sample Article Using IEEEtran.cls for IEEE Journals}
	\begin{abstract}
		\textcolor{black}{In this paper, we introduce a family of novel sparse array designs called the augmented ULAs (AULAs) for the localization of non-circular signals (NCS). \textcolor{black}{Accurate direction of arrival (DOA) estimation and the ability to resolve multiple targets are critical in modern wireless communication systems,} where antenna arrays must provide a larger co-array aperture, increased degrees of freedom (DOFs), and weaker mutual coupling. \textcolor{black}{However, most existing sparse arrays are optimized solely for the difference co-array, making them less efficient at utilizing the sum co-array resulting from the non-zero pseudo-covariance of NCS. Meanwhile, state-of-the-art designs for joint optimization of the sum and difference co-arrays remain constrained by a three-way performance trade-off.} The proposed AULAs configure single sparse and two dense ULAs alongside two separate elements to achieve a perfect splicing of holes and lags in the difference and sum co-array. This results in a larger virtual aperture and increased DOFs for NCS. Building on this structure, other variants of AULAs are developed, each exhibiting distinct characteristics. The shifted AULAs (SAULAs) judiciously displace the AULAs structure to minimize co-array redundancy and further enhance the DOFs. A transformed SAULAs (TSAULAs) design is proposed, which mitigates mutual coupling effects by converting the dense ULAs of SAULAs into sparse ULAs. By reconfiguring the elements of TSAULAs, the complementary TSAULAs (Co-TSAULAs) design inherits the desirable properties of SAULAs and TSAULAs.} \textcolor{black}{All these structures belong to a unified design framework, within which one configuration can be adapted into another during the design phase to meet different performance requirements.} Meanwhile, they provide in-built physical locations for convenient extension to a larger aperture. Closed-form expressions for precise element placements, DOFs, and weight functions are derived. Simulation results validate the effectiveness of the proposed approach.
	\end{abstract} 
	
	\IEEEpeerreviewmaketitle
	\section{Introduction}
	
	\IEEEPARstart {T}{h}e evolution of modern wireless communications and vehicular networks relies on accurate direction of arrival (DOA) estimation \cite{a, b}. By delivering high-resolution DOA estimation, base stations can accurately locate the users \cite{d},   directly assisting interference suppression and enabling robust beamforming for multi-user MIMO and vehicle-to-everything communication \cite{b, c}. Hence, the demand for superior angular resolution and capability of resolving more targets than the number of physical sensors is growing rapidly in next-generation wireless systems \cite{e}.
	
Typically, a larger aperture and higher degrees of freedom (DOFs) \textcolor{black}{ directly improve} the resolution and capability to detect multiple targets, with the array geometry being crucial to these enhancements. \textcolor{black}{While uniform linear arrays (ULAs) are popular due to their simple array structure, practical challenges arise in wireless communication systems because achieving more DOFs and a larger array aperture require a proportional increase in the number of sensors, which can increase the hardware cost, weight, and power consumption in practice \cite{SWAP1, SWAP2,  SWAP3}.} Meanwhile, the half wavelength spacing $(d = \lambda/2)$ between the elements of ULAs causes strong mutual coupling (MC), which degrades resolution performance \cite{UADF, SWAP1}. \textcolor{black}{Hence, it becomes imperative to develop improved array structures or signal processing algorithms to achieve high-resolution DOA estimation under these practical system constraints \cite{DCA2}.}

	To address these challenges, sparse arrays (or non-uniform linear arrays (NULAs)) offer a promising solution since they can generate a larger aperture and more DOFs than the ULAs \cite{6}. \textcolor{black}{Generally, sparse arrays are designed to leverage the difference co-array $(\mathbb{DC})$\cite{NA, DCA1}.} The minimum redundancy arrays (MRAs) generate more virtual elements with a wider $\mathbb{DC}$ aperture \cite{8}. However, they lack closed-form design expressions \cite{9, GMRA}. 
	
	Co-prime arrays (CPAs) \cite{CPA} and nested arrays (NA) \cite{NA} possess simple closed-form design expressions. Based on the difference co-array concept, different modifications of NA and CPA have been introduced to increase the DOFs \cite{SFCA, ktimes, DCA2, SNA, SpNA, ENA, TSCPA, Interpolation1, Interpolation2}. The NA combines dense and sparse ULAs to produce a hole-free $\mathbb{DC}$ with higher DOFs \cite{NA}. However, almost half of its elements are densely packed, making the NA more susceptible to MC. The CPA concatenates two subarrays that satisfy coprimality for inter-element spacing and the number of elements \cite{CPA, ECPA}. Although the CPA alleviates MC effects, its $\mathbb{DC}$ is not hole-free. This not only reduces the spatial efficiency (S.E) and number of uniform DOFs (uDOFs) that various DOA estimators rely on for estimation \cite{NA}, \cite{ktimes}, \textcolor{black}{but also wastes valuable aperture resources.} \textcolor{black}{While the information lost due to holes can be recovered via interpolation algorithms, such techniques are not feasible for real-time wireless systems where computational resources are limited.}
	
	In \cite{MISC}, a new design based on maximum inter-element spacing constraint (MISC) is introduced. The MISC arrays produce more DOFs and a hole-free $\mathbb{DC}$, while the MC effects are comparable to the CPA. \textcolor{black}{By combining multiple ULAs and additional sensors, SA-4U1S and SA-4U2S sparse array configurations were proposed to mitigate the MC while maintaining sufficient uDOFs \cite{DCA1}.} Although the pentad-displaced ULAs \textcolor{black}{(Pdis-ULAs)} achieve more DOFs than the MISC and SA-4U1S arrays \cite{Pdis}, they induce stronger MC.
	
	\textcolor{black}{Importantly, the aforementioned designs configure the array sensors primarily to optimize the $\mathbb{DC}$ derived from the second‑order covariance matrix. However, many signals used in wireless communication systems are non-circular (e.g., minimum shift keying, binary phase shift keying, pulse amplitude modulation, etc.), which provide additional information through a pseudo‑covariance (elliptic) matrix, generating a sum co-array ($\mathbb{SC}$) \cite{NonzeroECM2,NADiS, OCA}. Consequently, these designs are inefficient at capitalizing on the information contained in the non-zero pseudo-covariance matrix, resulting in a fragmented sum-difference co-array $(\mathbb{SDC})$ for non-circular sources (NCS) \cite{NonzeroECM2, NADiS, OCA}.} This has motivated recent research into designing new arrays that jointly optimize the $\mathbb{SC}$ and $\mathbb{DC}$, to obtain a longer, ideally hole-free $\mathbb{SDC}$ for enhanced DOA estimation of NCS.
	
	The nested array with displaced subarrays (NADiS) increases the inter-element spacing of the sparse ULA and redefine the inter-ULA displacement of the NA \cite{NADiS}. While this lengthened the uDOFs and continuous virtual aperture (CVA) for NCS, the sensitivity to MC remains the same as the NA. Likewise, the new nested array (NNA) also rearranges elements of NA to better utilize the $\mathbb{SDC}$ model \cite{NNA}. However, the NNA can only be designed when the array number is even. The concept of the nested array with graded spacing (GSNA) is proposed \cite{GSNA}, in which two elements, one between the two ULAs of prototype NA \cite{NA} and the other succeeding its sparse ULA, are appended. Although the GSNA obtains more DOFs than the NADiS, its $\mathbb{SDC}$ suffers from numerous holes \cite{GSNA}.
	
	In \cite{NSANCS}, a new sparse array for non-circular sources (NSANCS) is designed by relocating the dense and sparse ULAs of the NA. Though the NSANCS produces a hole-free virtual array, it offers almost the same number of DOFs as GSNA.
	\textcolor{black}{By reconfiguring the elements in the sparse ULA of the prototype NA, an optimized sparse nested array (OSNA) \cite{OSNA} is introduced to optimize the $\mathbb{SDC}$. However, the OSNA also contain a long dense ULA in its structure, making it as susceptible to the MC effects as the NA.} The sparse array for non-circular sources (SANC) \cite{SANC} provides a large virtual array with enhanced DOFs. However, designing SANC requires exhaustive search algorithms \cite{SANC}. \textcolor{black}{The zero-redundancy sparse array (ZRSA) \cite{ZRSA} produces a hole-free $\mathbb{SDC}$ and possesses closed-form expressions. Nevertheless, a strong MC results from a dense subarray in its physical array. Although the MC effects can be mitigated by rearranging the dense subarray, it can only be designed for $N$ being a multiple of 4.}
	
	On the other hand, the super transformed nested array (STNA) suppresses the MC effects \cite{STNA} and is free from design constraints. However, the $\mathbb{SDC}$ it generates is not a hole-free ULA. 
	\textcolor{black}{Relocation of sparse nested arrays (RSNA) is studied in \cite{RSNA}, which is based on incorporating two subarrays and a separate element. Though the RSNA increases the DOFs, these improvements are relatively moderate.} The triad-displaced ULAs (Tdis-ULAs) design \cite{Tdis} produces a hole-free $\mathbb{SDC}$ with higher DOFs capacity than the ZRSA and STNA, but their MC is stronger than that of the STNA \cite{STNA}. 
	
	Different derivatives of the CPA are also introduced for NCS.
	The special co-prime array (SCPA) shifts one of the co-prime subarrays to obtain a longer hole-free segment in $\mathbb{SDC}$ \cite{SCPA}. By rearranging the subarrays of CPA, the enhanced co-prime array (EnCPA) \cite{EnCPA} further improves the uDOFs. \textcolor{black}{An optimized shifted co-prime array with sum-difference co-array (OCA-SDCA) \cite{OCA} is proposed, which alleviates the MC and increases the DOFs compared to the EnCPA.  \textcolor{black}{Recently, a novel thinned coprime nested array (TCNA) is proposed \cite{TCNA}, which embeds coprime structure into the dense subarray of the NA. Although TCNA alleviates the MC effects and generates more uDOFs than OCA-SDCA, its uDOFs remain smaller than those of the STNA for the same number of sensors. Also, the construction of TCNA imposes parity restrictions on choosing the coprime integers, which limits its flexibility. Meanwhile, its $\mathbb{SDC}$ contains numerous holes. This is disadvantageous, particularly in compact wireless communications systems, where space is limited and aperture resources must be fully exploited.}}
	
	\textcolor{black}{Therefore, despite various sparse arrays designs proposed to jointly optimize $\mathbb{SC}$ and $\mathbb{DC}$, they remain constrained by a three-way performance trade-off. Designs like ZRSA \cite{ZRSA}, NSANCS \cite{NSANCS}, and Tdis-ULAs \cite{Tdis} produce high uDOFs with a hole-free $\mathbb{SDC}$, but experience strong MC. Conversely, \textcolor{black}{TCNA} \cite{TCNA}, OCA-SDCA \cite{OCA}, STNA \cite{STNA}, and RSNA \cite{RSNA} alleviate MC effects but lose the hole-free co-array property, with varying impacts on uDOFs. Meanwhile, OSNA \cite{OSNA}, NADiS \cite{NADiS}, and GSNA \cite{GSNA} suffer from both strong MC and several holes in the $\mathbb{SDC}$.}	
	
	\textcolor{black}{In this paper, we address this trade-off by introducing a novel augmented ULAs (AULAs) design framework for NCS. Unlike existing approaches, the AULAs framework presents a unified family of array designs that collectively address three core issues: (i) enabling a hole-free sum-difference co-array with an extended virtual aperture, (ii) significantly increasing uDOFs, and (iii) effectively reducing MC effects.} The prototype AULAs structure consists of two dense ULAs and a single sparse ULA alongside two separate elements. By strategically relocating the elements, other variants of AULAs family can be easily designed, each owing unique features. Specifically, the main contributions are summarized below: 
	\begin{itemize}
		\item The AULAs produce a hole-free $\mathbb{SDC}$ with higher uDOFs than the existing designs.
		\item The shifted AULAs (SAULAs) are introduced to minimize the redundancy between the $\mathbb{SC}$ and $\mathbb{DC}$ \textcolor{black}{by applying a unique, optimal shift} to the AULAs structure during the design phase. This results in an increase in DOFs and virtual aperture without requiring additional elements, while still maintaining a hole-free $\mathbb{SDC}$.
		\item By transforming the dense ULAs of SAULAs into sparse ULAs and effectively relocating an element, the transformed SAULAs (TSAULAs) structure is devised. TSAULAs significantly reduce the MC effects while achieving more uDOFs than AULAs. Although TSAULAs generate two holes, these occur only at the edges of $\mathbb{SDC}$, leaving a long hole‑free ULA intact.
		\item A complementary TSAULAs (Co-TSAULAs) design is also introduced, which inherits the properties of SAULAs and TSAULAs to produce a hole-free  $\mathbb{SDC}$ while being more robust against MC than SAULAs.
		\item All structures within the unified AULAs framework are easily interchangeable at the design stage and feature in-built physical locations for convenient aperture extension.
		\item The proposed AULAs family enjoy closed-form expressions for element locations, weight functions, and uDOFs.
	\end{itemize}
	
	The rest of the paper is organized as follows. 
	Necessary preliminaries are presented in Section II. Section III introduces the AULAs and SAULAs structures, including their array designing and co-array properties. Likewise, section IV discusses the TSAULAs and Co-TSAULAs designs. Section V presents the simulation results to demonstrate the excellent merits of the proposed design framework. Section VI concludes the paper.
	
	The bold upper-case (lower) characters throughout this manuscript show the matrices (vectors). The superscripts $(.)^*$, $(.)^T$, and $(.)^H$ respectively symbolize the conjugation, transpose, and conjugate transpose. $ \otimes $ signifies the Kronecker product. $E[.]$ and $\%$ denote the statistical expectation and remainder, respectively. \textbf{I} represents the identity matrix and vec$(.)$ converts a matrix into vector form. \textcolor{black}{$\rm{diag}(\textbf{a})$ is a diagonal matrix with its diagonal element obtained from $\textbf{a}$}. \textcolor{black}{$|.|$ returns the cardinality of set.} ${\left\| . \right\|_F}$ is the frobenius norm.  $\left\lceil . \right\rceil $ is ceiling function, which rounds to the nearest integer towards positive infinity. $\left\langle {u_{1}, u_{2}}\right\rangle$ signifies integer set $\{u \in \mathbb{Z}| u_{1} \le u \le u_{2} \}$.
	%
	%
	%
	%
	\vspace{-1em}\section{Preliminaries}
	\subsection{Data Model}
	Consider \textit{Z} narrowband and uncorrelated sources from  directions $\theta_z (z \in \left\langle {1, Z} \right\rangle)$ impinge on an \textit{N}-element array. The element locations are defined by $m_{n}d$, where $m_{n}$ indices belong to $\mathbb{P} = \{m_{n}, n = 0,1,\cdots, N \textcolor{black}{- 1}\}$. Then, the received signal vector $\textbf{s}(t)$ for \textit{t}th snapshot is modeled as \cite{NADiS}
	\begin{equation} \label{O.vector}
		\textbf{s}(t) = \textbf{A}\textbf{x}(t) + \textbf{n}(t), ~~~~~~~~~~~~~ t = 1,2,\cdots, T
	\end{equation}
	where $\textbf{x}(t) = [x_1, x_2, \cdots, x_Z]$ and $\textbf{n}(t)$ represents the signal waveform vector and noise vector of $p_n$ variance, respectively. ${\textbf{A}} = \left[ {{\textbf{a}}({\theta _1}), \textbf{a}({\theta _2}), \cdots, {\textbf{a}}({\theta _Z})} \right]$ is the ${N\times Z}$ direction matrix, with ${\textbf{a}}({\theta _z}) = \left\{ {e^{{-jl{m_n}{\rm{sin}{\theta _\textit{z}}}}}},{m_n} \in \mathbb{P} \right\}$ corresponds to the steering vector for $z$th NCS and ${l} = 2\pi d/\lambda$. The covariance matrix corresponding to $\textbf{s}(t)$ can be expressed as \cite{NADiS, OSNA}
	\begin{equation}
		\begin{array}{l} \label{cov}
			\begin{array}{l}
				{{\textbf{R}}_{\rm{s}}} = E[{\textbf{s}}(t){{\textbf{s}}^H}(t)] = {\textbf{A}}{{\textbf{R}}_{\rm{x}}}{{\textbf{A}}^H} + {p_n}{{\textbf{I}}_N},\\
				\begin{array}{*{20}{c}}
					{}&{ = \sum\limits_{z = 1}^Z {{p_z}{\textbf{a}}({\theta _z}){\textbf{a}^H}({\theta _z})} }
				\end{array} + {p_n}{{\bf{I}}_N},
			\end{array}
		\end{array}
	\end{equation}
	where ${\textbf{R}_{\rm{x}}} = {\rm{diag}}(\{p_1, p_2, \cdots, p_z\})$, with $p_z$ is the \textit{z}th source power. 
	\textcolor{black}{Notably, the NCS additionally contains a non-zero pseudo covariance matrix, which can be defined as \cite{NADiS, OCA}
		\begin{equation} \label{Pcov}
			\begin{array}{l}
				\begin{array}{*{20}{l}}
					{{{\hat {\textbf{R}}}_{\rm{s}}} = E[{\textbf{s}}(t){{\textbf{s}}^T}(t)] = {\textbf{A}}{{\hat {\textbf{R}}}_{\rm{x}}}{{\textbf{A}}^T},}
				\end{array}\\
				\begin{array}{*{20}{c}}
					{}&{ = \sum\limits_{z = 1}^Z {{{\hat p}_z}\textbf{a}({\theta _z}){\textbf{a}^T}({\theta _z})} }
				\end{array},
			\end{array}
		\end{equation}
		where ${\hat {\textbf{R}}_{\rm{x}}} = {\rm{diag}}(\{{\hat p_1}, {\hat p_2}, \cdots, {\hat p_z}\})$. \textcolor{black}{${{\hat {\textbf{R}}}_{\rm{s}}}$ is a complex symmetric pseudo‑covariance matrix that characterizes NCS and is zero for circular signals \cite{SANC, NADiS, Tdis, OCA}. This ${{\hat {\textbf{R}}}_{\rm{s}}}$ provides additional information for NCS, which can be exploited to achieve performance enhancements in DOA estimation \cite{NADiS, OCA}. Hence, the exploitation of pseudo-covariance is inherent to the data model for NCS \cite{NADiS}, \cite{GSNA,NSANCS,OSNA,SANC,ZRSA,STNA,RSNA,Tdis,SCPA,EnCPA,OCA}.}
		Subsequently, an extended output vector can be formed as \cite{OCA, NADiS, Tdis}
		\begin{equation}
			{\textbf{s}_o}(t) = [{\textbf{s}^T}(t), {\textbf{s}^H}(t)]\textcolor{black}{^T}.   
		\end{equation}
		The covariance matrix of ${\textbf{s}_{\rm{o}}}(t)$ can be expressed as \cite{OCA, NADiS}
		\begin{equation} \label{toextendB}
			{\textbf{R}_{{\rm{s}}_{\rm{o}}}} = {{E}}[{{{\rm{\textbf{s}}}}_{\rm{o}}}(t){\textbf{s}_{\rm{o}}}^H(t)] = \left[ {\begin{array}{*{20}{c}}
					{{{\bf{R}}_{\rm{s}}}}&{{{{\bf{\hat R}}}_{\rm{s}}}}\\
					{{{{\bf{\hat R}}}^ * }_{\rm{s}}}&{{\bf{R}}_{\rm{s}}^ * }
			\end{array}} \right].
	\end{equation}}
	Vectorizing ${\textbf{R}_{{\rm{s}}_{\rm{o}}}}$ yields a $4N^2 \times 1$ vector defined as \cite{OCA, NADiS}
	\begin{equation} \label{vec}
		\begin{array}{l}
			{\textbf{z}_s}_{o} \!=\! \textcolor{black}{{\rm{vec}}}({{\bf{R}}_{{{\rm{s}}_{\rm{o}}}}}) \!=\! \textbf{K}\left[ \begin{array}{l}
				{\textbf{B}^ + }\left( {{\theta _1},{\theta _2}, \cdots ,{\theta _z}} \right){\textbf{r}}_{\rm{x}} + {p_n}\textcolor{black}{{\rm{vec}}}({I_N})\\
				{\textbf{C}^ - }\left( {{\theta _1},{\theta _2}, \cdots ,{\theta _z}} \right)\hat {\textbf{r}}_{\rm{x}}^ * \\
				{\textbf{C}^ + }\left( {{\theta _1},{\theta _2}, \cdots ,{\theta _z}} \right){{\hat {\textbf{r}}}_{\rm{x}}}\\
				{\textbf{B}^ {\textcolor{black}{-}} }\left( {{\theta _1},{\theta _2}, \cdots ,{\theta _z}} \right)\hat {\textbf{r}}_{\rm{x}}^ *  + {p_n}\textcolor{black}{{\rm{vec}}}({I_N})
			\end{array} \right],
		\end{array}  
	\end{equation}
	where \textcolor{black}{\( \textbf{K} = (\textbf{I}_2 \otimes \textbf{P} \otimes \textbf{I}_N)^{-1} \) is a \( 4N^2 \times 4N^2 \) permutation matrix, 
		with \( \textbf{P} = \textbf{D}_{mn} \otimes \textbf{D}_{mn}^T \in \mathbb{R}^{(2N \times 2N)} \), 
		and \( \textbf{D}_{mn} \) is a \( N \times 2 \) elementary matrix having unit element in the \((m,n)\)-th location 
		(for \( m = 1,2,\cdots,N \) and \( n = 1,2 \)).}
	
	In (\ref{vec}), ${\textbf{r}}_{\rm{x}} = \rm{diag}({\textbf{R}}_{\rm{x}})$ and ${{\hat {\textbf{r}}}_{\rm{x}}} =  \rm{diag}({{\hat {\textbf{R}}}_{\rm{x}}})$ signify the steering vectors of the covariance and pseudo covariance matrices, respectively, while $\textbf{B}^ +, \textbf{C}^ -, \textbf{C}^ +, \textbf{B}^ -$ are $N^2 \times Z$ matrices defined as \cite{OCA, NADiS, OSNA} 
	\begin{equation}
		\!\!\left[ {\begin{array}{*{20}{l}}
				\!\!{{{\textbf{B}}^ + }\left( {{\theta _1},{\theta _2}, \cdots ,{\theta _z}} \right)}\\ \!\!
				{{{\bf{C}}^ - }\left( {{\theta _1},{\theta _2}, \cdots ,{\theta _z}} \right)}\\\!
				\!{{{\bf{C}}^ + }\left( {{\theta _1},{\theta _2}, \cdots ,{\theta _z}} \right)}\\\!
				\!{{{\bf{B}}^ {\textcolor{black}{-}} }\left( {{\theta _1},{\theta _2}, \cdots ,{\theta _z}} \right)}
		\end{array}} \!\!\! \right] \!\! \! = \! \!\!\left[ {\begin{array}{*{20}{c}}
				\!\!\!\!\!\!\!\!\!{{\textbf{a}^ * }({\theta _1}) \!\otimes \! \textbf{a}({\theta _1}), \cdots\!, {\textbf{a}^ * }({\theta _Z}) \!\otimes \!\textbf{a}({\theta _Z})}\\\!\!\!\!
				{{\textbf{a}^ * }({\theta _1}) \!\otimes\! {\textbf{a}^ * }({\theta _1}), \cdots \!, {\textbf{a}^ * }({\theta _Z}) \!\otimes \!{\textbf{a}^ * }({\theta _Z})}\\\!\!\!\!\!\!\!\!\!\!\!\!\!\!\!
				{\textbf{a}({\theta _1}) \!\otimes\! \textbf{a}({\theta _1}),\cdots\!, \textbf{a}({\theta _Z}) \!\otimes \!\textbf{a}({\theta _Z})}\\\!\!\!\!\!\!\!\!\!\!
				{\textbf{a}({\theta _1}) \!\otimes \!{\textbf{a}^ * }({\theta _1}),\cdots\!, \textbf{a}({\theta _Z}) \!\otimes \!{\textbf{a}^ * }({\theta _Z})}
		\end{array}} \!\!\!\!\right].
	\end{equation}
	\textcolor{black}{Here, $\mathbf{B}^- = (\mathbf{B}^+)^*$ and $\mathbf{C}^- = (\mathbf{C}^+)^*$ indicates that these matrix pairs are conjugate‑symmetric.}
	It is easy to observe that the first $N^2$ and last $N^2$ entries of ${\textbf{z}_s}_{_o}$, $[({\textbf{a}^ * }({\theta _z}) \otimes {\textbf{a}}({\theta _z}))^T, ({\textbf{a}}({\theta _z}) \otimes  {\textbf{a}^ *}({\theta _z}))^T]^T$ emulates the virtual steering vector of ${x_z}(t)$ in the $\mathbb{DC}$, which can be shown as \cite{OCA, NADiS}
	\begin{equation} \label{DC}
		\mathbb{DC} = {\rm{ }}\left\{ {{m_u} - {m_v}\mid {m_u}, {m_v} \in \mathbb{P}} \right\}.
	\end{equation}
	Meanwhile, the second and third entries of ${\textbf{z}_s}_{o}$, $[({\textbf{a}^ * }({\theta _z}) \otimes {\textbf{a}^ * }({\theta _z}))^T, ({\textbf{a}}({\theta _z}) \otimes  {\textbf{a}}({\theta _z}))^T]^T$ represent the virtual steering vector of ${x_k}(t)$ in the $\mathbb{SC}$ defined by
	\begin{equation} \label{SC}
		\mathbb{SC}  = \pm {\rm{ }}\left\{ {{m_u} + {m_v}\mid {m_u}, {m_v} \in \mathbb{P}} \right\}.
	\end{equation}
	Based on (\ref{DC}) and (\ref{SC}), we analyze that the vector ${\textbf{z}_s}_{o}$ can be viewed as the signal received by the $\mathbb{SDC}$, defined as:
	\begin{equation} \label{RCS}
		\mathbb{SDC} \! = \!\left\{ {{m_u} \!- \!{m_v}, \!-{m_u} \!- \!{m_v},  {m_u} \!+\! {m_v}| {m_u}, {m_v} \in \mathbb{P} } \right\}.
	\end{equation}
	
	Most of the existing DOA estimators exploit the consecutive part of the $\mathbb{SDC}$ to estimate the DOAs of NCS \cite{NA, Coarray, NADiS, OCA}. \textcolor{black}{The $\mathbb{SDC}$ model is well-established and central to the DOA estimation of NCS \cite{NADiS}, \cite{GSNA,NSANCS,OSNA,SANC,ZRSA,STNA,RSNA,Tdis,SCPA,EnCPA,OCA}. 
		For a fair performance comparison, all array designs in this manuscript, including \cite{NA},\cite{MISC},\cite{Pdis}, are evaluated based on $\mathbb{SDC}$.}
	\textcolor{black}{\begin{definition}(Difference Set): Consider $\mathbb{P_1} = \{ {m_{\textcolor{black}{u}}}|\textcolor{black}{u} = 1,2, \cdots ,M\}$, $\mathbb{P_2} = \{ {m_{\textcolor{black}{v}}}|\textcolor{black}{v} = 1,2, \cdots ,N\}$ are the two locations sets, then the difference set between them can be defined as  $\mathbb{P_2} - \mathbb{P_1} = \{ {m_{\textcolor{black}{v}}} - {m_{\textcolor{black}{u}}}| {m_{\textcolor{black}{v}}} \in \mathbb{P_2}, {m_{\textcolor{black}{u}}} \in \mathbb{P_1}\}.$ 
		\end{definition}
		\begin{definition}(Sum Set): Consider two location sets as, $\mathbb{P_1} = \{ {m_{\textcolor{black}{u}}}|\textcolor{black}{u} = 1,2, \cdots ,M\}$, $\mathbb{P_2} = \{ {m_{\textcolor{black}{v}}}|\textcolor{black}{v} = 1,2, \cdots ,N\}$, then their sum set can be defined as  $\mathbb{P_2} + \mathbb{P_1} = \{ {m_{\textcolor{black}{v}}} + {m_{\textcolor{black}{u}}}| {m_{\textcolor{black}{v}}} \in \mathbb{P_2}, {m_{\textcolor{black}{u}}} \in \mathbb{P_1}\}.$ 
		\end{definition}
		\begin{definition} (Continuous Virtual Aperture): \textcolor{black}{Based on (\ref{RCS}), the virtual array aperture can be quantified as $max(\mathbb{SDC}) - min(\mathbb{SDC})$ \cite{New}}. Let \textit{m} be the largest integer such that $\left\{ {0, \pm 1, \cdots, \pm m} \right\} \subseteq \mathbb{SDC}$. Then, the continuous virtual segment is defined as $\mathbb{U} = \left\{ {0, \pm 1, \cdots, \pm m} \right\}$, which has no missing lags between $-m$ and $m$. \textcolor{black}{The CVA, representing the spatial extent of this segment, is calculated for a unit inter-element spacing $d$ as:
				\begin{equation}
					CVA = max(\mathbb{U}) - min(\mathbb{U}) = 2m.
			\end{equation}}.
		\end{definition}
		\vspace{-2em}\begin{definition}
			(Uniform Degrees of Freedom): The uDOFs are the number of distinct lags in $\mathbb{U}$, \textcolor{black}{given by:
				\begin{equation}
					uDOFs = |\mathbb{U}| = 2m + 1.
			\end{equation}}
		\end{definition} \vspace{-2.25em}
		\begin{definition}
			(Spatial Efficiency): The spatial efficiency (S.E) is the ratio of the length of consecutive segment in the positive  $\mathbb{SDC}$ to the total virtual length of the positive $\mathbb{SDC}$ \cite{ktimes}.
		\end{definition}
		\begin{definition}
			Denote ${\mathbb{F}}$ as the filled difference or sum co-array with no missing lags. For instance, the filled $\mathbb{DC}$ can be defined as ${\mathbb{F}} = \left\{ {n} \in {\mathbb{Z}}| {\min (\mathbb{DC}) \le n \le \max (\mathbb{DC})} \right\}.$ Then, the integer is a missing lag or a hole (h) in $\mathbb{DC}$ if $h \in \mathbb{F}$ but $h \notin \mathbb{DC}.$ While, the set of holes $(\mathbb{DH})$ is formulated by $\mathbb{DH} = \mathbb{F} \backslash \mathbb{DC} = \left\{h \in \mathbb{F}|~ h \notin \mathbb{DC}  \right\}$. Likewise, it applies to the case of $\mathbb{SC}$, for which the set of holes is defined by $\mathbb{SH} = \mathbb{F} \backslash \mathbb{SC} = \left\{h \in \mathbb{F}|~ h \notin \mathbb{SC}  \right\}$, where ${\mathbb{F}} = \left\{ {n} \in {\mathbb{Z}}| {\min (\mathbb{SC}) \le n \le \max (\mathbb{SC})} \right\}.$
	\end{definition}}  \vspace{-2em}
	\subsection{Mutual Coupling}
	The MC between the elements with smaller inter-element separation distorts the estimation performance. The output vector in the presence of MC is formulated by \cite{STNA}
	\begin{equation}
		\textbf{x}(t) = \textbf{C}\textbf{A}\textbf{s}(t) + \textbf{n}(t), ~~~~~~~~~~~~~ t = 1,2,\cdots, T
	\end{equation}
	where \textbf{C} signifies the $N \times N$ MC matrix, which for the coupling-free models is an identity matrix. \textbf{C} can be approximated by a \textit{B}-banded Toeplitz matrix, expressed by:
	\begin{equation} \label{MCM}
		{\left[ \textbf{C} \right]_{{m_u},{m_v}}} = \left\{ \begin{array}{l}
			{c_{\left| {{m_u} - {m_v}} \right|}}, ~~~~ {\rm{if}}\left| {{m_u} - {m_v}} \right| \le B,\\
			0, ~~~~~~~~~~~~~ {\rm{otherwise,}}
		\end{array} \right.
	\end{equation}
	where ${{m_u}, {m_v}} \in \mathbb{P}$, ${\left[ \textbf{C} \right]_{{m_u},{m_v}}}$ signifies the $({m_u},{m_v})$th entry of \textbf{C} and the coupling coefficients ${c_0},{c_1}, \cdots, {\rm{ }}{c_B}$ satisfy ${c_0} = 1 \textcolor{black}{\ge} |{c_1}| > |{c_2}| >  \cdots  > |{c_B}|$. The weight functions and coupling leakage are essential to quantify the MC effects.
	\begin{definition}
		The weight functions $w(f)$ of an array ($\mathbb{P}$) are the number of element pairs corresponding to the co-array index \textit{f}, which can be expressed as \cite{STNA,MISC}
		\begin{equation}
			w(f) = {\rm{ }}|\left\{ {\left( {{m_1},{m_2}} \right) \in {\mathbb{P}}^2|{m_u} - {m_v} = f} \right\}|, ~~~~ f \in {\mathbb{DC}}
		\end{equation}
	\end{definition}
	\begin{definition}
		The coupling leakage for an \textit{N}-element array can be characterized as the energy ratio \cite{MISC}
		\begin{equation}
			L_{C} = \frac{{{{\left\| {\textbf{C} - {\rm{diag}}(\textbf{C})} \right\|}_F}}}{{{{\left\| \textbf{C} \right\|}_F}}}.
		\end{equation}
		A Higher value of $L_{C}$ indicates stronger MC. 
	\end{definition}
	\vspace{-1.5em}\section{Proposed Augmented ULAs (AULAs) Design}
	\begin{figure}	[!t] \vspace{-0.5em}\hspace{-0.25em}
		\includegraphics[width = 0.5\textwidth]{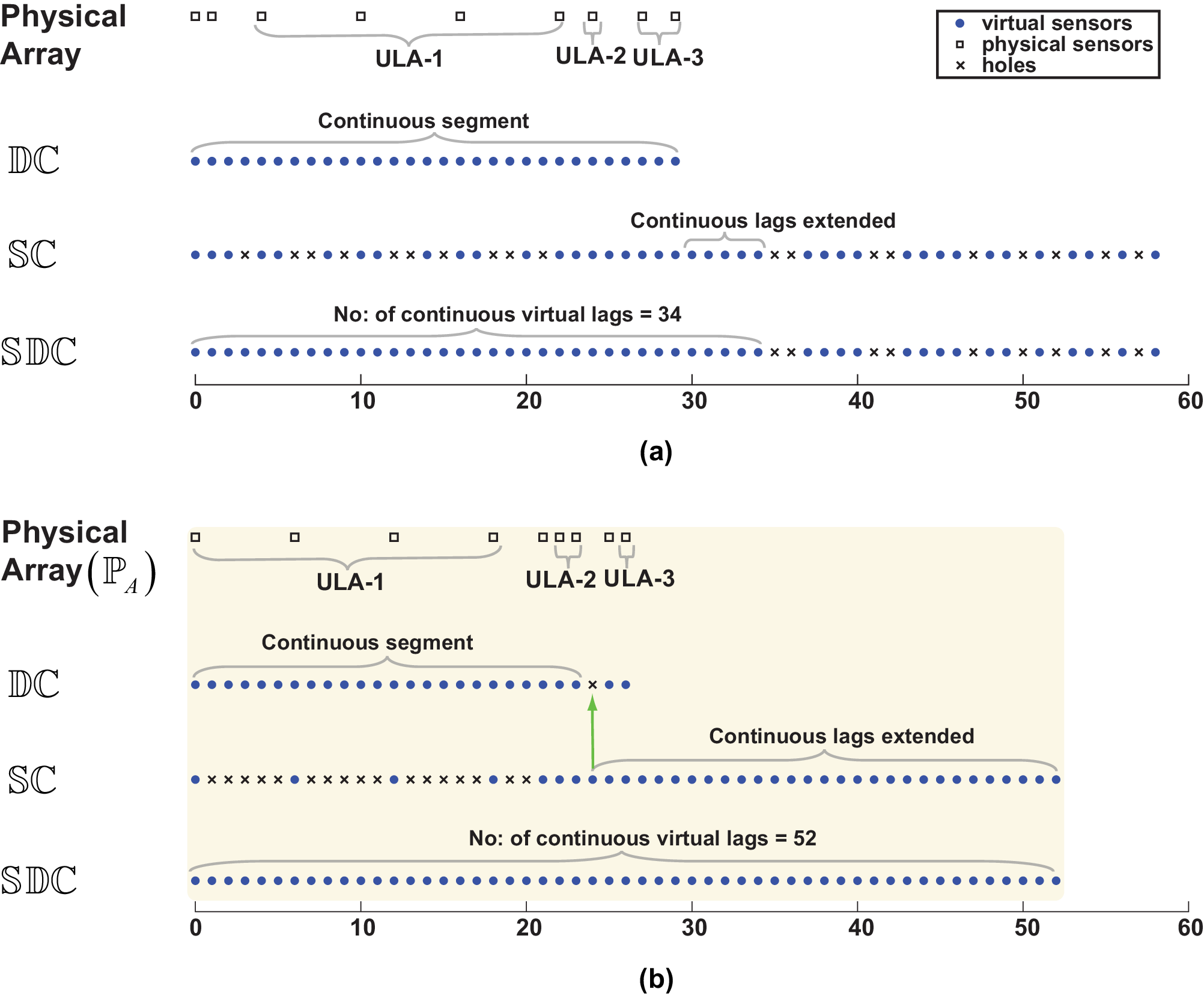} \vspace{-0.9em}
		\caption{\textcolor{black}{\textcolor{black}{Design progression from MISC to the AULAs for $N = 9$. (a) Physical and virtual arrays of MISC. (b) Physical and virtual arrays of AULAs.}} \label{MISCvsAULA}} \vspace{-1.5em}
	\end{figure}
	\subsection{Motivation}
	\textcolor{black}{Fig.\,\ref{MISCvsAULA} illustrates the design progression from MISC arrays to the proposed AULAs configuration.} \textcolor{black}{The MISC array contains three sparse ULAs and two separate elements. The inter-element spacing in two of the sparse ULAs is fixed to 2\textit{d}, and the maximum inter-element spacing of $Md$ is set for the remaining ULA. \textcolor{black}{While the MISC achieves a large hole‑free $\mathbb{DC}$ with this array structure, it is suboptimal for the $\mathbb{SDC}$.}}
	
	Fig.\,\ref{MISCvsAULA}(a) plots the physical and virtual (non‑negative) arrays (normalized by \textit{d}) of the MISC for $N = 9$. The MISC array produces a long hole‑free $(\mathbb{DC})$ up to 29, where the maximum inter‑element spacing is $M = 6$ \cite{MISC}. \textcolor{black}{However, for NCS, which involve the joint utilization of both $\mathbb{DC}$ and $\mathbb{SC}$ to form the $\mathbb{SDC}$, MISC exhibits structural shortcomings. Its $\mathbb{SC}$ contains irregular holes and contributes marginally to extending the continuous virtual lags beyond the $\mathbb{DC}$. Consequently, the $\mathbb{SDC}$ of MISC is continuous only up to 34. This not only restricts uDOFs but also wastes aperture resources, as the holes congest the central portion of $\mathbb{SDC}$.}
	
	
	\textcolor{black}{To overcome these limitations, we introduce the AULAs design. As shown in Fig.\,\ref{MISCvsAULA}(b), AULAs employ a novel configuration that integrates three ULAs and two separate elements to jointly optimize the $\mathbb{DC}$ and $\mathbb{SC}$. Here, two ULAs have an inter‑element spacing $d$, the third ULA uses spacing $Md$, and the two separate elements are uniquely positioned. \textcolor{black}{AULAs systematically avoid holes inherent to the $\mathbb{SDC}$ of MISC array.}}
	
	Although the maximum inter-element spacing of AULAs is also $M = 6$ for $N = 9$, and its continuous $\mathbb{DC}$ produces lags up to 23, \textcolor{black}{AULAs achieve a critical advantage in Fig.\,\ref{MISCvsAULA}: a virtual lag in its $\mathbb{SC}$ fills the hole in $\mathbb{DC}$ (indicated by a green arrow)}. Beyond this splicing of holes and virtual lags, the remaining part of $\mathbb{SC}$ is also continuous up to 52. \textcolor{black}{In contrast, the $\mathbb{SC}$ of MISC remains fragmented.} As a result, the AULAs generate a hole-free $\mathbb{SDC}$, with virtual lags up to 52.
	\vspace{-1.5em}\subsection{Design Rules And Array Structure of AULAs}
	The total number of elements of AULAs is obtained by $N$ = $N_{1} + N_{2} + N_{3} + 2$, where $N_{1}, N_{2}, N_{3}$ respectively signify the number of elements in the first, second, and third ULA. Fig.\,\ref{AULA} illustrates the AULAs configuration. We use the associated inter-element spacing set $\mathbb{A}_{\rm{AULAs}}$ and maximum inter-element spacing $(M)$ to determine the element locations of AULAs. Specifically, $M$ and $\mathbb{A}_{\rm{AULAs}}$ can be defined as
	\begin{equation} \label{PV}
		\begin{array}{l}
			M = 2\left\lceil {N/4} \right\rceil,  ~~~ \textcolor{black}{(N \ge 9)} \\
		\end{array} 
	\end{equation} 
	\begin{equation} \label{AULAs-A}
		\mathbb{A}_{\rm{AULAs}} = \{ \underbrace{M, M, \cdots, M}_{	{N_1} - 1}, M/2, \underbrace{1, 1, \cdots, 1}_{	{N_2}}, 2, \underbrace{1, 1, \cdots, 1}_{	{N_3}}\},
	\end{equation}
	where
	\begin{equation} \label{PV1}
		\left\{ \begin{array}{l}
			{N_1} = N - M{\rm{ + 1}},\\
			{N_2} = M/2 - 1{\rm{,}}\\
			{N_3} = M/2 - 2{\rm{.}}
		\end{array} \right.
	\end{equation}
	\textcolor{black}{Note that, since $N_3 = M/2 - 2$, ensuring that all three ULAs contain at least one element requires $M \ge 6$, which in turn implies that $N \ge 9$. This guarantees that the AULAs structure forms three ULAs with two separate sensors.}
	
	\begin{figure} \vspace{-1.75em}\centering
		\includegraphics[width = 0.47\textwidth]{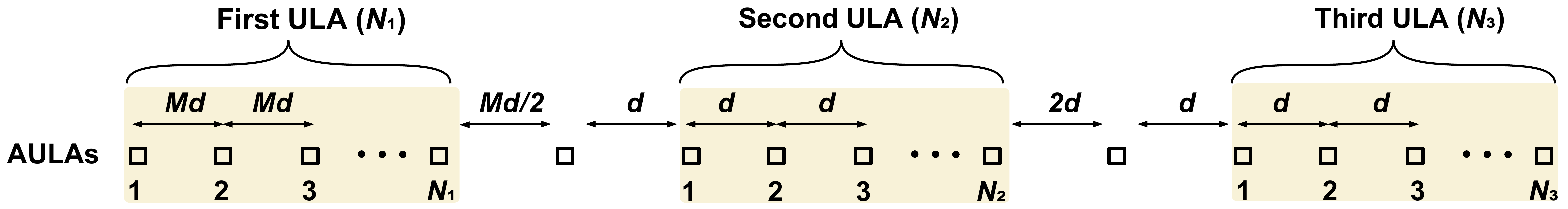} \vspace{-1.1em}
		\caption{Augmented ULAs configuration.} \label{AULA} \vspace{-1.5em}
	\end{figure}
	\begin{figure}	[b]  \centering
		\includegraphics[height=3.55cm,width=7.5cm]{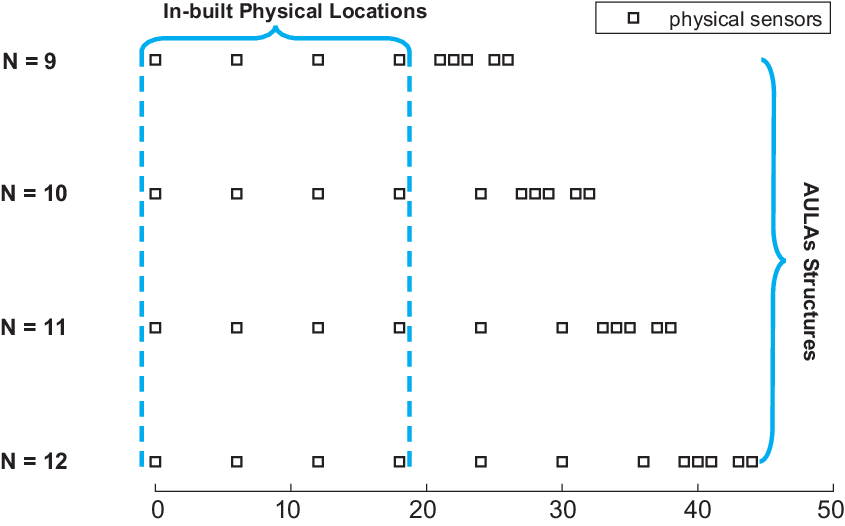}
		\vspace{-0.8em}	\caption{\textcolor{black}{Physical sensor locations for AULAs $(N \in \left\langle {9,12} \right\rangle)$, illustrating the in-built location property.} \label{inbuilt}}
	\end{figure}
	Following (\ref{AULAs-A}), the array configuration $({\mathbb{P}_{\rm{A}}})$ of AULAs can be mathematically defined as
	\begin{equation} \label{AULAA}
		{\mathbb{P}_{\rm{A}}} = {\mathbb{P}_{\rm{a_1}}} \cup {\mathbb{P}_{\rm{a_2}}} \cup {\mathbb{P}_{\rm{a_3}}}, 
	\end{equation}
	where 
	\begin{equation} \label{AULAA1}
		{\mathbb{P}_{\rm{a_1}}} = \{\left\langle {0,({N_1} - 1)} \right\rangle (M) \},
	\end{equation}
	\begin{equation} \label{AULAA2}
		{\mathbb{P}_{\rm{a_2}}} = \{({N_1}M - M/2 - 1) + \left\langle {1, M/2} \right\rangle\},
	\end{equation}
	\begin{equation} \label{AULAA3}
		{\mathbb{P}_{\rm{a_3}}} = \{({N_1}M) + \left\langle {1, M/2 - 1} \right\rangle\}.
	\end{equation}
	\begin{lemma}
		The AULAs enjoy the following properties: \\
		(a) The $\mathbb{DC}$ is continuous in $ \left\langle {0, {l_{a_1}}} \right\rangle$, except a single hole occur at ${N_1}M$, where ${l_{a_1}} = {N_1}M + M/2 - 1.$ \\
		(b) The $\mathbb{SC}$ is continuous in  $ \left\langle {{l_{a_2}}, {l_{a_3}}} \right\rangle$, where
		\textcolor{black}{\begin{equation}
				\left\{ {\begin{array}{*{20}{l}}
						{l_{a_2}} = {N_1}M - M/2,\\
						{l_{{a_3}}} = {2{N_1}M + M - 2.}
				\end{array}} \right.
		\end{equation}} \\
		(c) The uDOFs for AULAs $({D_A})$ is obtained by
		\textcolor{black}{	\begin{equation}
				{D_A} = {\begin{array}{l}
						4{N_1}M + 2M - 3.
				\end{array}} 
		\end{equation}}
	\end{lemma} \vspace{-2em}
	\begin{proof}
		The proof is provided in Appendix A.
	\end{proof} \vspace{-1.5em}
	\subsection{In-built Physical Locations for Scaling}
	Another key merit of the AULAs design is its in-built physical locations, enabling easier scaling to a larger aperture. The AULAs maintain $M - 2$ identical locations for any array number in an interval, whose first and last array number reside in $N\%4 = 1$ and $N\%4 = 0$, respectively. \textcolor{black}{For example, when $M = 6$, the interval corresponds to $N \in \langle 9, 12 \rangle$.} 
	
	\textcolor{black}{Fig.\,\ref{inbuilt} illustrates the in-built physical location property of AULAs for $N \in \langle 9, 12 \rangle$, where $M = 6$. As shown, all the structures have $M - 2 = 4$ identical locations. Thus, a 9-element AULAs design inherently provides four sensor locations for the other configurations. Consequently, extending it to a 10-element structure requires configuring only the remaining $N - M + 2 = 6$ locations, reducing the redesign effort by 40\%. Notably, when starting from a larger base array within the same interval, the number of identical locations increases. For example, a 10-element AULAs structure yields five of the sensor locations required for an 11-element AULAs, reducing the redesign effort by approximately 45\%. Likewise, an 11-element AULAs provide six sensor locations for a 12-element AULAs, reducing the redesign workload by 50\%. Since the value of $M$ increases with $N$, the scalability advantage become more pronounced for larger arrays. Therefore, the AULAs design approach facilitates scaling to a larger aperture with minimal restructuring of the physical array.} \vspace{-1em}
	\subsection{Weight Functions}
	The weight functions $w(f)$ are critical for quantifying the MC effects. Specially, the first three weight functions, $w(1)$, $w(2)$, and $w(3)$ dominate the coupling behavior, with $w(1)$ exerting the highest impact \cite{STNA}.
	
	\begin{figure}	[t] \vspace{-0.5em} \hspace{-0.25em} 
		\includegraphics[width = 0.5\textwidth]{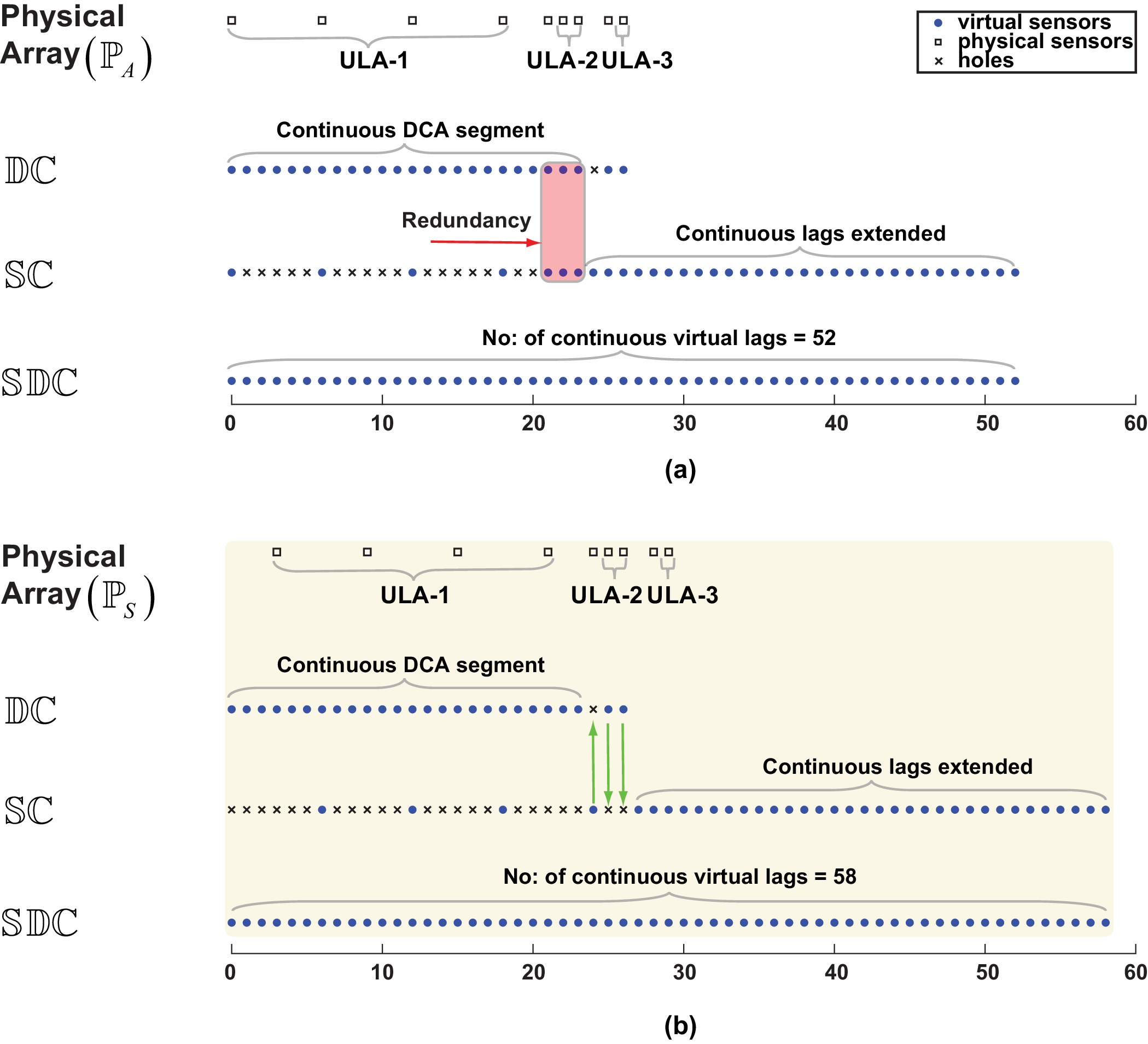}
		\vspace{-1.5em}	\caption{\textcolor{black}{Design progression from AULAs to SAULAs configuration using $N = 9$ as an example. (a) Physical and virtual arrays of AULAs. (b) Physical and virtual arrays of SAULAs.} \label{Augtoshifted}}
	\end{figure}
	According to the $\mathbb{A}_{\rm{AULAs}}$ in (\ref{AULAs-A}) and the definition of $w(f)$, the number of element pairs corresponding to 1, 2, 3 in ${\mathbb{A}_{\rm{AULAs}}}$ respectively yields $w(1), w(2), w(3)$. Therefore, the first three weight functions for AULAs can be formulated as
	\begin{align} \label{WAULAs-A}
		\scalemath{0.9}{\hspace{-0.5em} w(1) \!= \!2\left\lceil {N/4} \right\rceil \! - \!3,w(2) \!=\! 2\left\lceil {N/4} \right\rceil \! - \!4,w(3) \!=\! \left\{\! \!{\!\!\begin{array}{*{20}{l}}
					{2\left\lceil {N/4} \right\rceil \! - \!5,}\!\!\!\!&{{\rm{if}}\;M \!> \!6,}\\
					{\left\lceil {N/4} \right\rceil ,}\!\!\!\!&{{\rm{if}}\;M \!= \!6.}
				\end{array}\!\!} \right.}
	\end{align}
	
	\subsection{Shifted Augmented ULAs (SAULAs) With Increased DOFs}
	The AULAs generate a long hole‑free $\mathbb{SDC}$. However, it is noted that \textcolor{black}{some virtual lags in the $\mathbb{DC}$ and $\mathbb{SC}$ overlap at identical locations}, resulting in redundant information. This is illustrated in Fig.~\ref{Augtoshifted}(a), which displays the physical and virtual (non-negative) arrays of the AULAs design. The red-shaded rectangle captures these overlapping virtual lags. Hence, by reducing this redundancy, we can further increase the DOFs.
	\begin{figure} [t]	\centering
		\includegraphics[width = 0.5\textwidth]{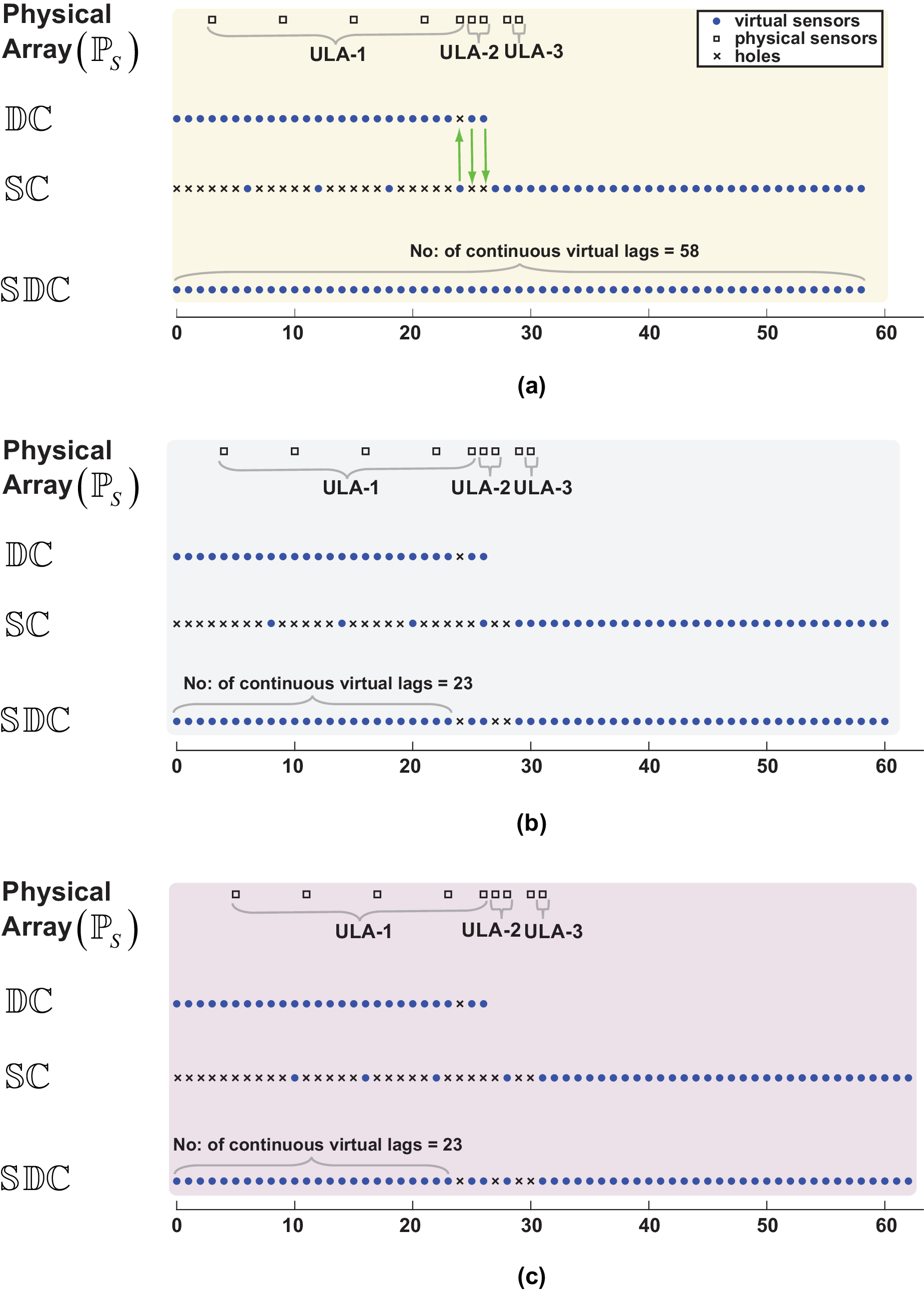} \vspace{-1.05em}
		\caption{\textcolor{black}{Step-wise analysis of arbitrary shifts applied to SAULAs for $N = 9$. (a) Physical and virtual arrays of SAULAs with $M/2$ shift. (b) Physical and virtual arrays of SAULAs with $M/2 + 1$ shift. (c) Physical and virtual arrays of SAULAs with $M/2 + 2$ shift.} \label{SAULAshift}} \vspace{-1.35em}
	\end{figure}
	
	To achieve this, we present the shifted AULAs (SAULAs) design by applying an optimal shift of $M/2$ to the AULAs structure. \textcolor{black}{This shift is implemented during the array design phase instead of an arbitrary post-design adjustment, resulting in a new, physically realizable array geometry that further optimizes the $\mathbb{SDC}$. As a result, while the $\mathbb{DC}$ remains identical to that of the AULAs, the lags of $\mathbb{SC}$ are now shifted by $M$, thereby minimizing the redundancy between $\mathbb{DC}$ and $\mathbb{SC}$.}  \textcolor{black}{Fig.~\ref{Augtoshifted}(b) shows the physical and virtual (non-negative) arrays of the SAULAs. Although both designs occupy the same total physical space, SAULAs achieves more DOFs and a larger CVA. \textcolor{black}{This is achieved through a two-way splicing strategy: the hole in $\mathbb{DC}$ at ${N_1}M = 24$ is filled by a virtual lag in the $\mathbb{SC}$, while the $\mathbb{DC}$ provides virtual lags that fill the subsequent holes in the $\mathbb{SC}$ at ${N_1}M + 1 = 25$ and ${N_1}M + 2 = 26$, as indicated by green arrows.} \textcolor{black}{This raises the number of consecutive virtual lags to 58, which is $M$ more than the AULAs design provides.}}
	
	\textcolor{black}{Furthermore, the applied shift $(M/2)$ is uniquely optimal, ensuring the $\mathbb{SDC}$ of SAULAs remains hole-free. Any increase beyond $M/2$ breaks this precise alignment and introduces holes. For instance, increasing the shift to $M/2 + 1$ prevents the hole in the $\mathbb{DC}$ at ${N_1}M$ from being filled by the $\mathbb{SC}$, causing $\mathbb{SDC}$ to no longer be a hole-free ULA. This is illustrated in Fig.\,\ref{SAULAshift}, which shows the physical and virtual (non-negative) arrays with different applied shifts, each resulting in distinct physical array, where $N\!=\!9, M=6,$ and $N_{1}=4$. When the shift is increased to $M/2 + 1 = 4$, a hole occurs at ${N_1}M = 24$, reducing the consecutive virtual lags by nearly 60\% compared to the applied shift of $M/2$. This is because, although the $\mathbb{DC}$ remains unchanged, the $\mathbb{SC}$ lags are shifted relative to the $\mathbb{DC}$ hole. Similarly, a further increase to $M/2 + 2 = 5$ lengthens the $\mathbb{SDC}$ but introduces more holes. Thus, $M/2$ is not an arbitrary shift but an optimal design parameter for SAULAs.}
	
	Fig.~\ref{SAULA} shows the SAULAs model. Since the inter-element spacings are identical for AULAs and SAULAs, their associated inter-element spacing set remain the same. Mathematically, the array configuration of SAULAs can be defined as
	\begin{equation} \label{SAULAA} 
		{\mathbb{P}_{\rm{S}}} = {\mathbb{P}_{\rm{s_1}}} \cup {\mathbb{P}_{\rm{s_2}}} \cup {\mathbb{P}_{\rm{s_3}}} ,
	\end{equation}
	where 
	\begin{equation} \label{SAULAA1} 
		{\mathbb{P}_{\rm{s_1}}} = \{ (M/2) + \left\langle {0,({N_1} - 1)} \right\rangle (M)\},
	\end{equation}
	\begin{equation} \label{SAULAA2} 
		{\mathbb{P}_{\rm{s_2}}} = \{({N_1}M - 1) + \left\langle {1, M/2} \right\rangle\},
	\end{equation}
	\begin{equation} \label{SAULAA3} 
		{\mathbb{P}_{\rm{s_3}}} = \{({N_1}M + M/2) + \left\langle {1, M/2 - 1} \right\rangle\},
	\end{equation}
	where ${N_1}$ and $M$ take the same values as in (\ref{PV}) and (\ref{PV1}).
	\begin{figure} [!t] \vspace{-1.5em}	\hspace{-0.5em}
		\includegraphics[width = 0.5\textwidth]{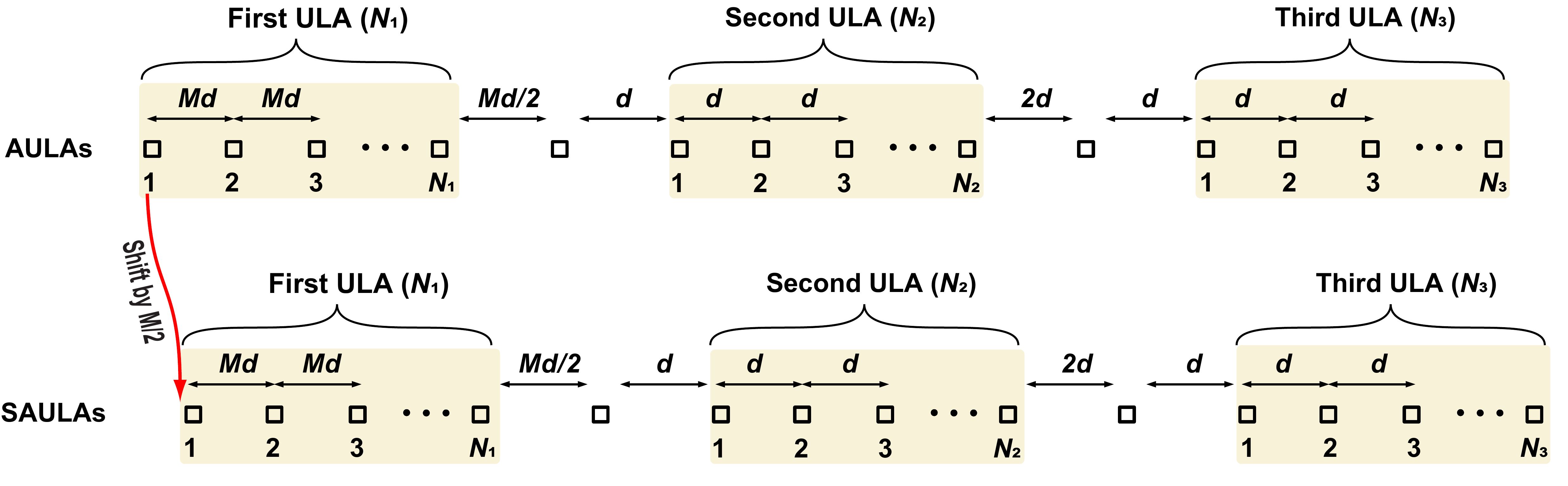} \vspace{-2.25em}
		\caption{Shifted AULAs configuration based on AULAs.} \label{SAULA}\vspace{-1em}
	\end{figure}
	\vspace{-1em}\subsection{\textcolor{black}{Comparison With Designs  Based on Shifting Technique}}
	
	\textcolor{black}{The OCA-SDCA \cite{OCA} reconfigures the CPA structure \cite{CPA} to better exploit the $\mathbb{SDC}$ model, resulting in increased uDOFs. Similar to CPA, OCA-SDCA also comprises two coprime subarrays containing $Q - 1$ and $2P$ sensors, where $P$ and $Q$ are coprime design parameters. The OCA-SDCA removes the leftmost sensor from the first subarray of the CPA, which is analogous to shifting the first subarray by $P$ and then removing its rightmost sensor. Meanwhile, the OCA-SDCA shifts second subarray to the right by $LQ$, where $0 \le L \le \lfloor \frac{{P + 1}}{2} \rfloor.$ Fig.\,\ref{OCAvsZRSAvsSAULAmii}(b) shows the synthesis of OCA-SDCA from the CPA structure (shown in Fig.\,\ref{OCAvsZRSAvsSAULAmii}(a)) for $N = 10$ with $P = 3$ and $Q = 5$, where the physical and virtual (non-negative) arrays are plotted. Clearly, the OCA-SDCA overlaps $\mathbb{SC}$ and $\mathbb{DC}$ more effectively than the CPA, yielding more continuous virtual lags.}
		\begin{figure} [t] \hspace{-0.25em}
		\includegraphics[height=11.7cm,width=8cm]{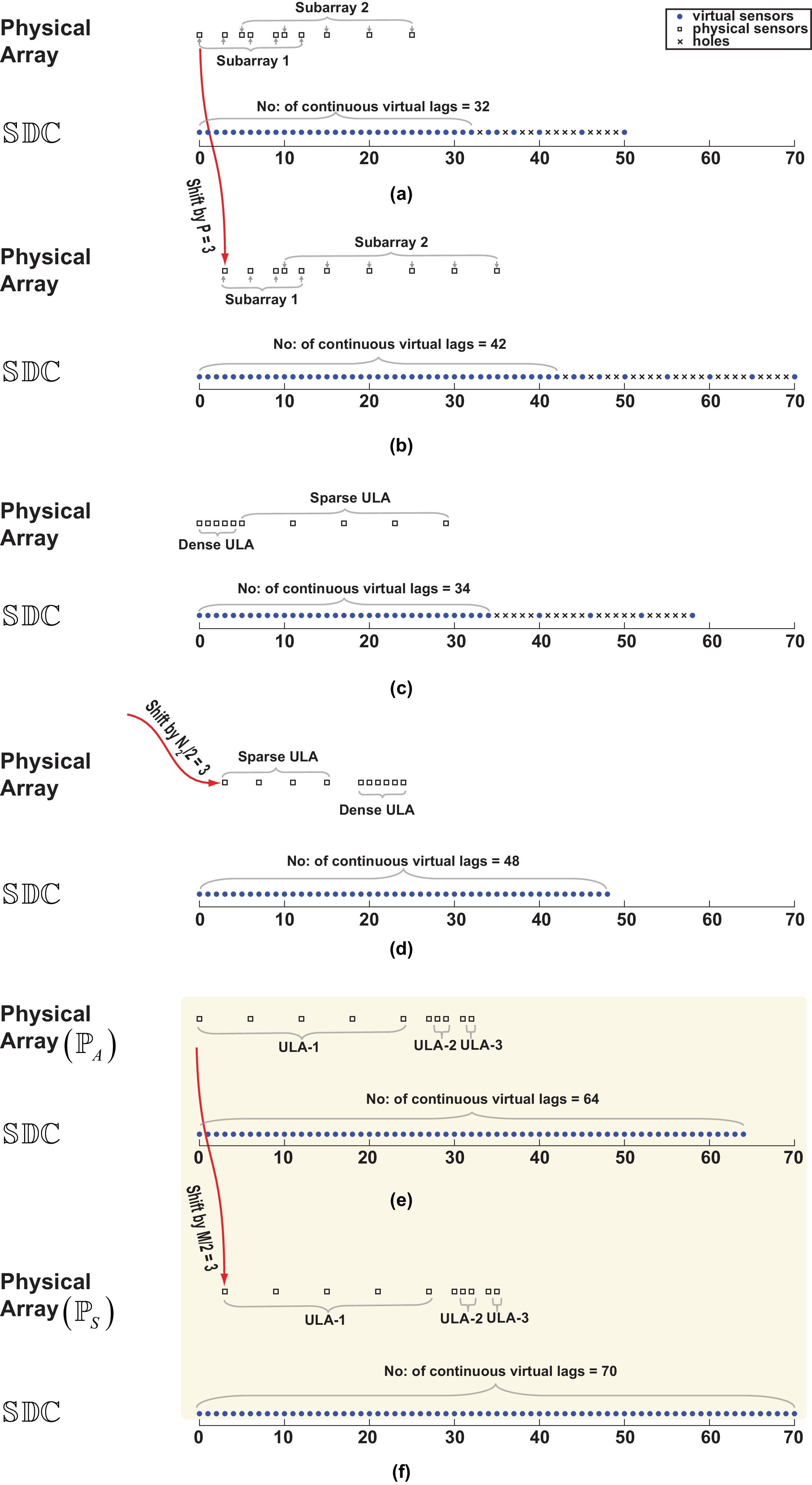} \vspace{-1.5em}
		\caption{\textcolor{black}{Evolution and comparison of different 10-element sparse array designs based on shifting technique. (a) CPA. (b) OCA-SDCA. (c) NA. (d) ZRSA. (e) AULAs. (f) SAULAs.} \label{OCAvsZRSAvsSAULAmii}} 
	\end{figure}
	
	\textcolor{black}{Likewise, ZRSA \cite{ZRSA} can be viewed as an evolution of the NA \cite{NA}. ZRSA strategically arranges the sparse and dense ULAs, with the sensor locations are shifted by ${N_2}/2$ to reduce the redundancy between $\mathbb{SC}$ and $\mathbb{DC}$. Here, ${N_2}$ is the number of sensors in the dense ULA of ZRSA. Fig.\,\ref{OCAvsZRSAvsSAULAmii}(d) shows the design evolution of ZRSA. Although the NA also comprises dense and sparse ULAs, it suffers from multiple holes. In contrast, the ZRSA produces a longer and hole-free $\mathbb{SDC}$.}
	
	\textcolor{black}{Fig.\,\ref{OCAvsZRSAvsSAULAmii}(f) illustrates the evolution of SAULAs from AULAs (shown in Fig.\,\ref{OCAvsZRSAvsSAULAmii}(e)) design. Unlike OCA-SDCA, SAULAs feature a hole-free $\mathbb{SDC}$, with the continuity exceeding that of the ZRSA. These improvements stem from the novel AULAs approach, enabling a more efficient splicing of $\mathbb{SC}$ and $\mathbb{DC}$ lags than ZRSA. The SAULAs further lengthen the $\mathbb{SDC}$ while retaining the hole-free  co-array property. In particular, SAULAs achieve approximately 46\% and 67\% more continuous virtual lags than the ZRSA and OCA-SDCA, respectively.}
	\vspace{-0.75em}\begin{lemma}
		The following properties hold for the SAULAs: \\
		(a) The $\mathbb{DC}$ is continuous in $  \left\langle {0, {l_{s_1}}} \right\rangle$, except one hole on $ {N_1}M$, where ${l_{s_1}} = {N_1}M + M/2 - 1.$ \\
		(b) The $\mathbb{SC}$ is continuous in  $ \left\langle {{l_{s_2}}, {l_{s_3}}} \right\rangle$, where 
		\textcolor{black}{	\begin{equation}
				\left\{ {\begin{array}{l}
						{l_{s_2}} = {N_1}M + M/2, \\
						{l_{s_3}} =	2{N_1}M + 2M - 2.
				\end{array} } \right.
		\end{equation}} \\
		(c) The uDOFs of SAULAs $({D_S})$ is formulated by
		\textcolor{black}{	\begin{equation}
				{D_S} = {\begin{array}{l}
						4{N_1}M + 4M - 3.
				\end{array} } 
		\end{equation}}
	\end{lemma} \vspace{-2em}
	\begin{proof}
		The proof is provided in Appendix B.
	\end{proof} 
	\begin{figure} [t]  \centering
		\includegraphics[width = 0.5\textwidth]{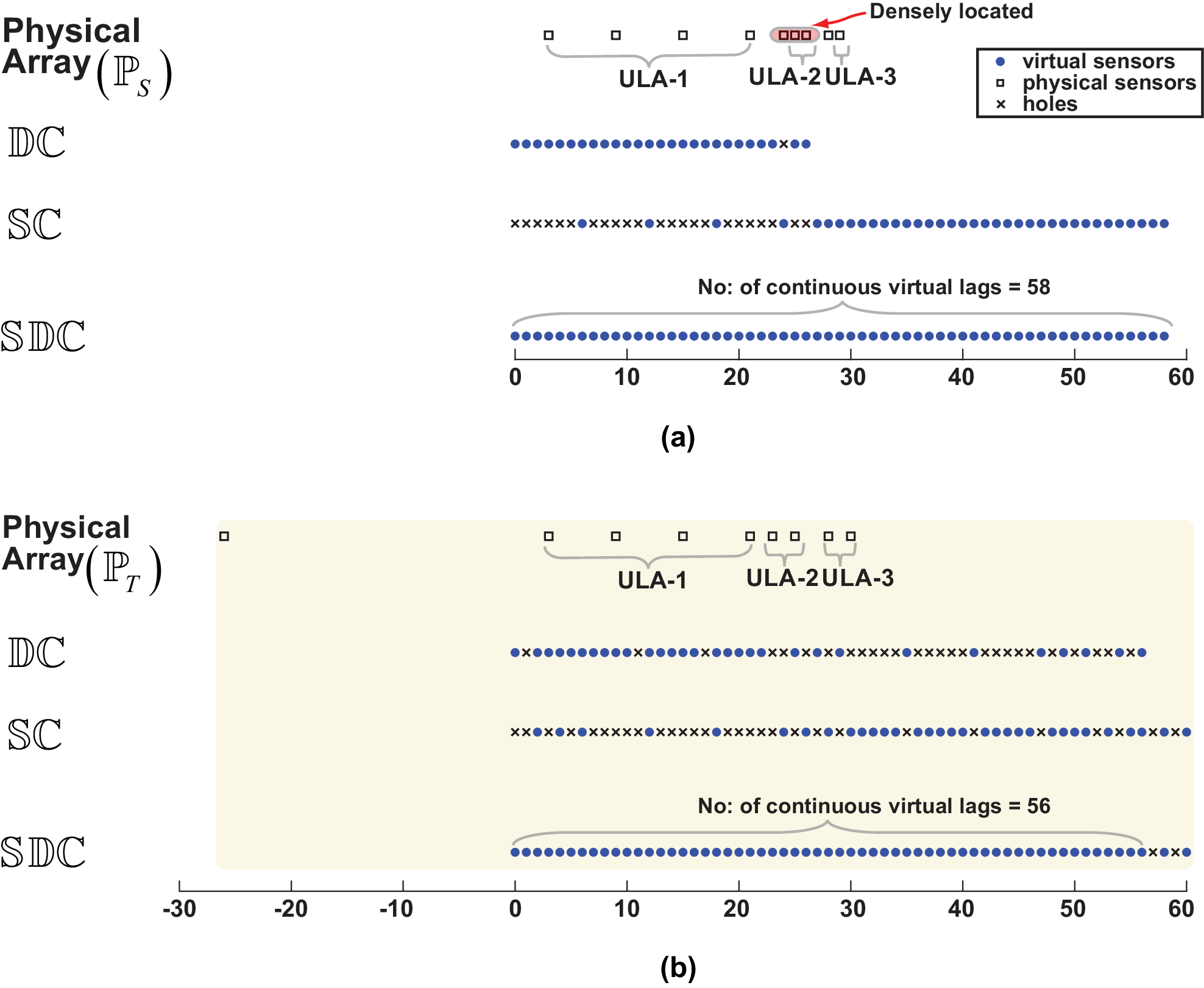} \vspace{-1.1em}
		\caption{\textcolor{black}{Design progression from SAULAs to TSAULAs for $N\!\!=\!\!9$. (a) Physical and virtual arrays of SAULAs. (b) Physical and virtual arrays of TSAULAs.} \label{ShifttoT}}
	\end{figure}
	\vspace{-1.25em}\begin{remark}
		SAULAs also feature in-built physical locations of the AULAs, as  they share similar structural layout. \textcolor{black}{However, the applied shift during the design phase defines a distinct set of sensor locations for SAULAs, as specified by (\ref{SAULAA}).} Likewise, this holds for their weight functions, since the associated inter-element spacing sets of the two designs are the same.
	\end{remark}
			
			\section{Transformed and Shifted Augmented ULAs (TSAULAs) Configuration}
			\textcolor{black}{Most of the existing designs, including \cite{NADiS}, \cite{NSANCS}, \cite{OSNA}, \cite{ZRSA},  \cite{RSNA}, \cite{Tdis} suffer from strong MC because a large number of their elements are placed with an inter-element spacing of $d.$ Although SAULAs have comparatively fewer separations of $d$, there can still be a performance loss in the presence of strong MC environment. Fig.~\ref{ShifttoT}(a) displays the physical and virtual (non-negative) arrays of SAULAs. As noticed, SAULAs contain three densely located element pairs ($d$ spacing) for $N\!\!= \!9$.} 
			\begin{figure} [t]	\centering
				\includegraphics[width = 0.5\textwidth]{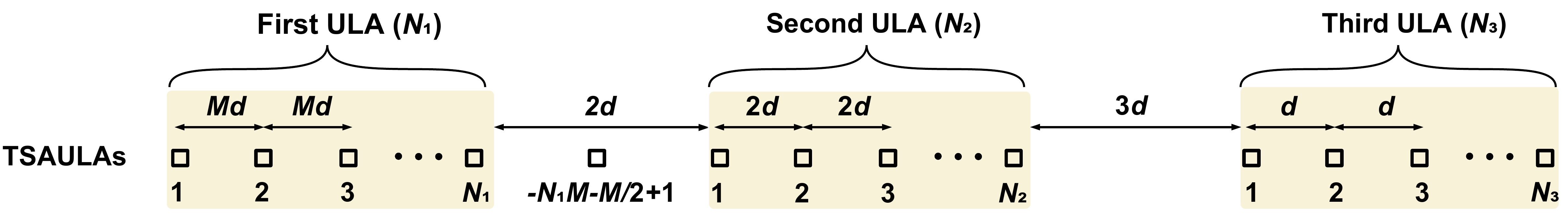} \vspace{-1.95em}
				\caption{Transformed and shifted augmented ULAs configuration.}  \label{TSAULA} \vspace{-1.6em}
			\end{figure}
			
			\vspace{-1em}This section introduces the TSAULAs design, which relocates the first single element of SAULAs to $(-N_{1}M - M/2 + 1)d$. Meanwhile, it transforms both dense ULAs of SAULAs into sparse ULAs, with each element spaced by 2\textit{d}. Also, the first two ULAs are separated by 2$d$, while the inter-ULA spacing for the second and third ULA is set to 3$d$. Subsequently, there is no spacing of $d$ in TSAULAs structure.
			
			\textcolor{black}{Fig.~\ref{ShifttoT}(b) shows the locations of physical and virtual (non-negative) arrays of TSAULAs. TSAULAs deplete the dense region in the SAULAs structure by eliminating all $d$-spacings. Despite having two holes in its $\mathbb{SDC}$, a long hole-free ULA remains intact as these holes appear at the tail.} Hence, TSAULAs produce only four fewer uDOFs than the SAULAs, but benefits from a sparser structure that inherently reduces MC.
			\vspace{-1.5em}\subsection{Design Rules and Array Structure}
			TSAULAs effectively configure three sparse ULAs plus a separate element, with the total number of array elements given by $N = N_{1} + N_{2} + N_{3} + 1$. Fig.~\ref{TSAULA} displays the proposed TSAULAs configuration. The associated inter-element spacing set $(\mathbb{A}_{\rm{TSAULAs}})$ and the maximum inter-element spacing set $(M)$ of TSAULAs can be expressed as
			\begin{equation} \label{set_transformeda}
				M = 2\left\lceil {N/4} \right\rceil, ~~~ \textcolor{black}{(N \ge 5)}\\
			\end{equation}
			\textcolor{black}{\begin{equation} \label{Associated_trans}
					\begin{array}{l} \mathbb{A}_\mathrm{TSAULAs} 
						= \{ \underbrace {M,M, \cdots ,M}_{{N_1} - 1},\underbrace {2,2, \cdots ,2}_{{N_2}}, - 2{N_1}M - M + 3, \\ 
						\quad \quad \quad \quad \quad \quad
						2{N_1}M + M,\underbrace {2,2, \cdots ,2}_{{N_3}-1}\},
					\end{array}
			\end{equation}}
			where
			\begin{equation} \label{set_transformed}
				\left\{ \begin{array}{l}
					{N_1} = N - M + 1,\\
					{N_2} = M/2 - 1{\rm{,}}\\
					{N_3} = M/2 - 1{\rm{.}}
				\end{array} \right.
			\end{equation}
			\textcolor{black}{ Since $N_2 \! = \!N_3 \!= \!M/2 \!- \!1$, ensuring that the TSAULAs design forms three ULAs and a separate sensor requires $N \ge 5$.}

			The location set $(\mathbb{P}_{\rm{T}})$ of TSAULAs is formulated by
			\begin{equation} \label{TSAULAA}
				\mathbb{P}_{\rm{T}} = {\mathbb{P}_{\rm{t_1}}} \cup {\mathbb{P}_{\rm{t_2}}} \cup {\mathbb{P}_{\rm{t_3}}},
			\end{equation}
			where
			\begin{equation} \label{TSAULAA1}
				{\mathbb{P}_{\rm{t_1}}} = \{(M/2) + \left\langle {0,({N_1} - 1)} \right\rangle (M)\},
			\end{equation}
			\begin{equation} \label{TSAULAA2}
				\begin{array}{l}
					{\mathbb{P}_{\rm{t_2}}} = {\mathbb{P}_{\rm{t_{2a}}}} \cup {\mathbb{P}_{\rm{t_{2b}}}},  \\ \quad   \quad \!\! =  \left\{ {({N_1}M - M/2) + \left\langle {1,M/2 - 1} \right\rangle ({\rm{2}})} \right\} \\
					\quad \quad \quad \cup  \left\{ {(-{N_1}M - M/2 + 1)} \right\},
				\end{array}
			\end{equation}
			\begin{equation} \label{TSAULAA3}
				{\mathbb{P}_{\rm{t_3}}} = \{({N_1}M + M/2 - 1) + \left\langle {1,M/2 - 1} \right\rangle ({\rm{2}})\}.
			\end{equation}
			\begin{figure}	[t] \vspace{-0.5em}
				\includegraphics[width = 0.5\textwidth]{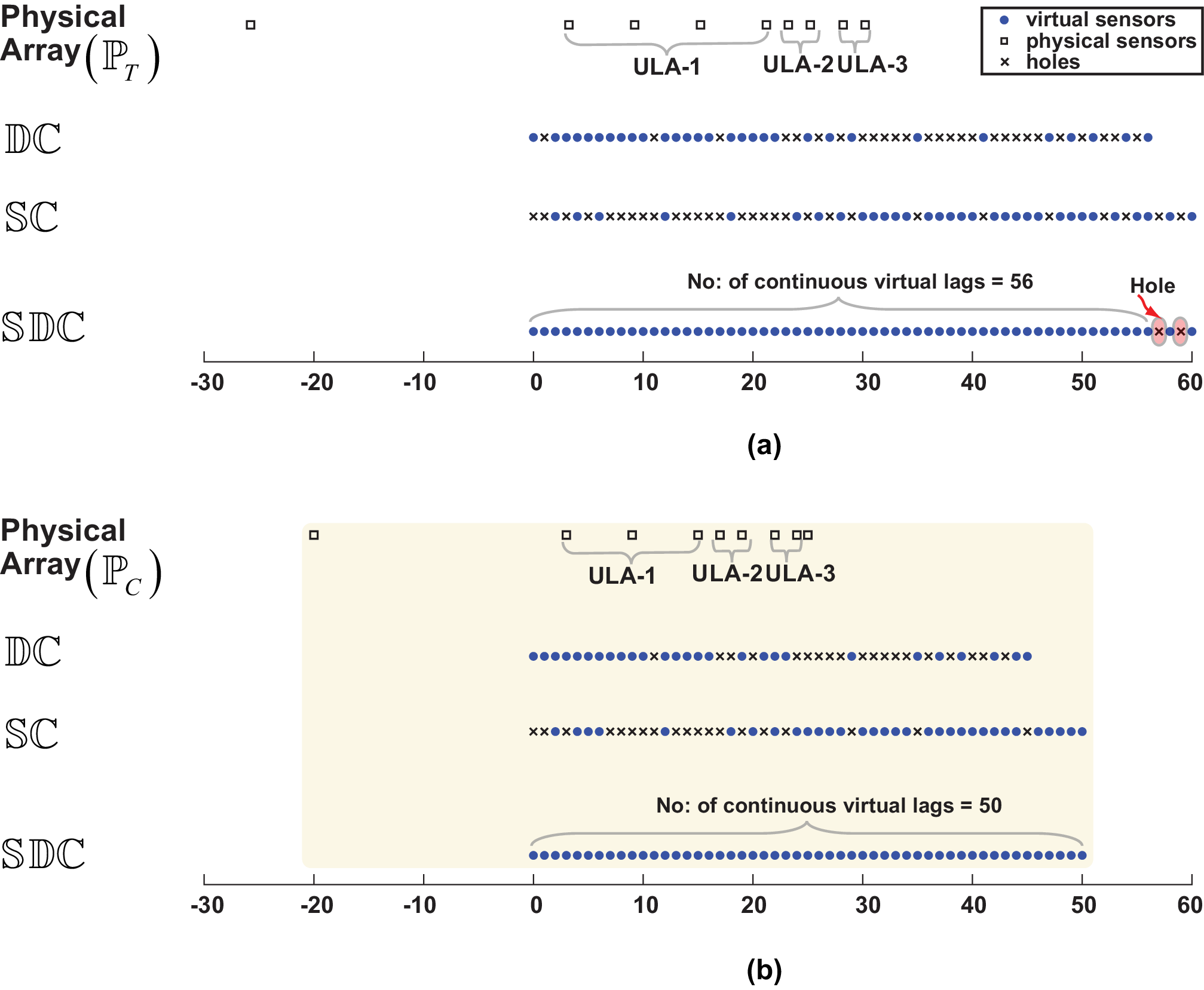} \vspace{-1.1em}
				\caption{\textcolor{black}{Design progression from TSAULAs to Co-TSAULAs configuration using $N = 9$ as an example. (a) Physical and virtual arrays of TSAULAs. (b) Physical and virtual arrays of Co-TSAULAs.} \label{TtoCo}} \vspace{-1.7em}
			\end{figure}
			\vspace{-2em}\begin{lemma}
				The TSAULAs enjoy the following properties:  \\
				(a) The joint sum-difference co-array is continuous in $ \pm  \left\langle {{l_{t_1}}, {l_{t_1}}} \right\rangle$, where ${l_{t_1}} = 
				2{N_1}M + 2M - 4.$ \\
				(b) The uDOFs for TSAULAs $({D_T})$ can be expressed as
				\begin{equation} \label{DTSAULA}
					{D_T} = 4{N_1}M + 4M - 7.
				\end{equation}
			\end{lemma}
			\begin{proof}
				Refer to Appendix C for the proof.
			\end{proof} 
			\begin{remark}
				Similar to SAULAs, the TSAULAs design also features in-built physical locations since the first ULA, containing these locations, remains identical for both structures. Hence, the TSAULAs provides an equal number of in-built physical locations as the SAULAs with similar conditions to satisfy.
			\end{remark}
			\begin{remark}
				The TSAULAs and SAULAs contain an equal number of elements in the first two ULAs, while the first ULA of the two designs is also identical. This implies that once the SAULAs are designed, it automatically provides the $N - M + 1$ locations of TSAULAs. Hence, the SAULAs can be easily adapted to TSAULAs during the design phase, while conversely, the reverse also holds. 
			\end{remark}
			\vspace{-1.5em}\subsection{Weight Functions}
			Based on (\ref{Associated_trans}), we can determine the weight functions of TSAULAs design at co-array index $f = 1,2,3$ as
			\begin{equation}
				w(1) = 0, w(2) = 2\left\lceil {N/4} \right\rceil - 3, {w(3)} = 1,
			\end{equation}
			where $w(1)$ and $w(3)$ are smaller than those of the SAULAs in (\ref{WAULAs-A}). Hence, it is conclusive that the TSAULAs are sparser than the SAULAs.
			\vspace{-1.5em}\subsection{Complementary Transformed \& Shifted Augmented ULAs (Co-TSAULAs) Configuration}
			This section presents the Co-TSAULAs design. \textcolor{black}{As discussed earlier, the SAULAs produce a hole-free $\mathbb{SDC}$ but contain two dense ULAs. While TSAULAs eliminate the dense ULAs, they introduce holes at the tail of $\mathbb{SDC}$. The Co‑TSAULAs are designed to inherit the valuable properties of the SAULAs and TSAULAs.} In particular, it both minimizes the number of closely-spaced element pairs compared to the SAULAs and produces a hole-free $\mathbb{SDC}$. The Co-TSAULAs achieve this by relocating the last element from the first ULA of the TSAULAs  to the tail of its third ULA with spacing \textit{d}.
			\begin{figure} [t] 	\centering 
				\includegraphics[width = 0.5\textwidth]{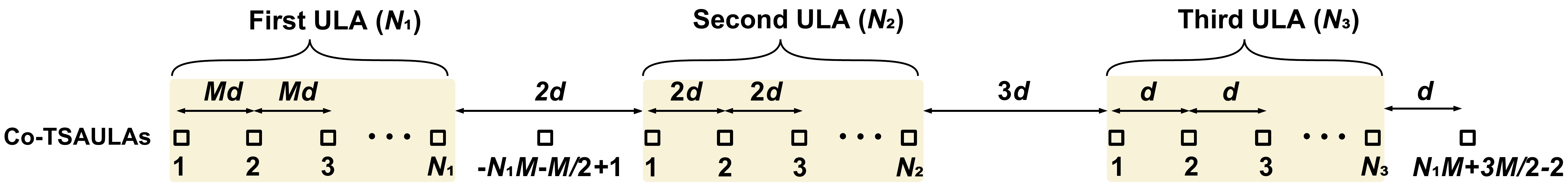} \vspace{-1.95em}
				\caption{Complementary transformed and shifted augmented ULAs design.} \label{CoTSAULA} \vspace{-1.4em}
			\end{figure}
			
			\textcolor{black}{Fig.~\ref{TtoCo} shows the physical and virtual (non-negative) arrays of the TSAULAs and Co-TSAULAs for $N = 9$. As observed, the lags in the $\mathbb{DC}$ of Co‑TSAULAs fill the holes in its $\mathbb{SC}$, while the reverse also holds true. Hence, the $\mathbb{SDC}$ of Co-TSAULAs is hole-free. 
				Meanwhile, the Co-TSAULAs contain fewer spacings of $d$ than the SAULAs (see in Fig.~\ref{ShifttoT}(a)).} 
			\begin{figure*} \vspace{-1.35em}\hspace{-7.5em}
				\includegraphics[height=10cm,width=24cm]{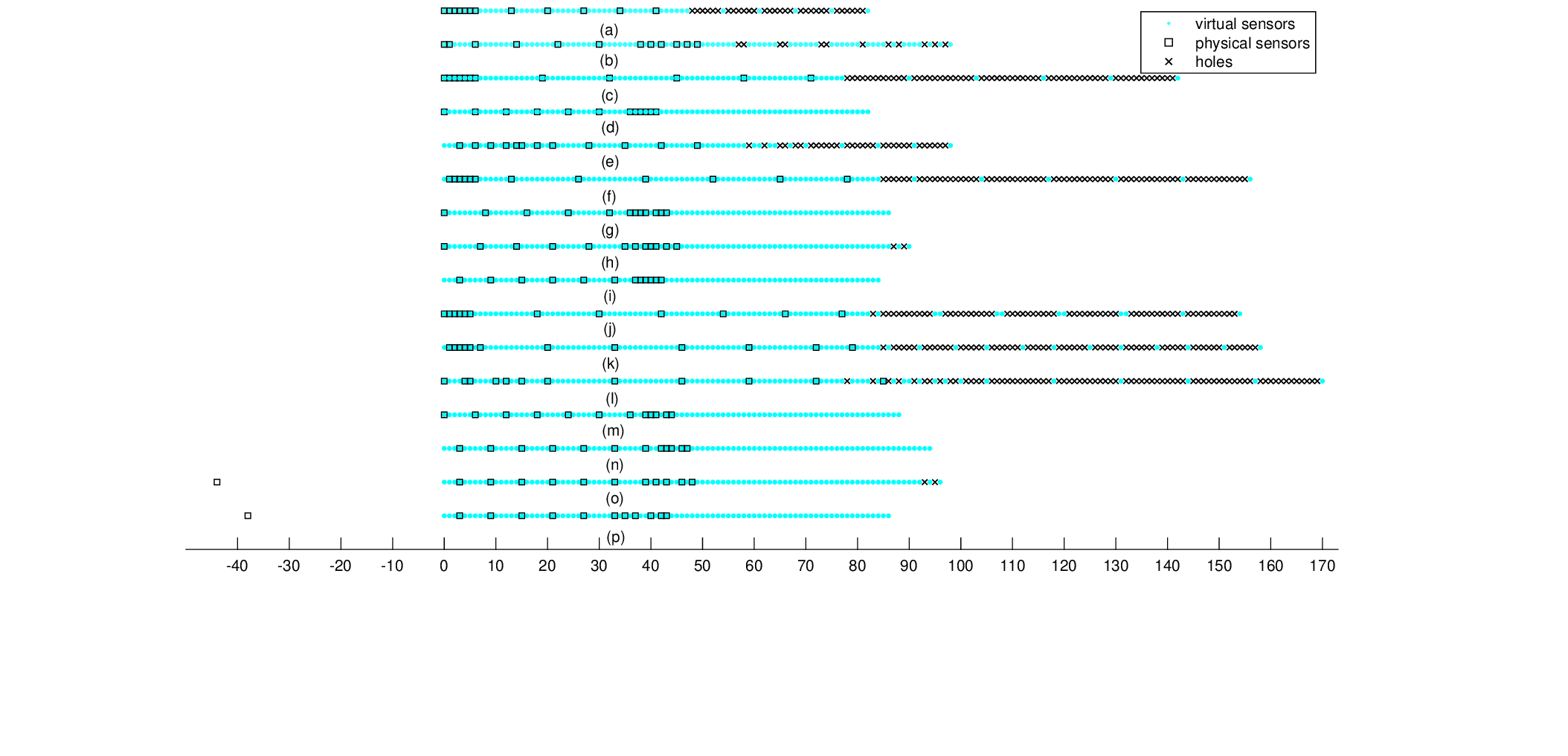} \vspace{-8.25em} 
				\textcolor{black}{\caption{Geometric distribution of physical and $\mathbb{SDC}$ \textcolor{black}{(non-negative)} elements for different sparse arrays. (a) NA. (b) MISC. (c) NADiS. (d) NSANCS. (e) OCA-SDCA. (f) OSNA. (g) Tdis-ULAs. (h) STNA. (i) ZRSA. (j) GSNA. (k) RSNA. (l) TCNA. (m) AULAs. (n) SAULAs. (o) TSAULAs. (p) Co-TSAULAs.} \label{Charac}} 
			\end{figure*}
			\begin{table*} [t!] \vspace{-1em}
				\scriptsize
				\caption{\textcolor{black}{Characteristics Comparison of Different Sparse Arrays}}
				\label{my-label}
				\setlength\tabcolsep{4pt}
					\begin{tabular}{P{1.2cm} P{0.4cm} P{0.7cm} P{0.7cm} P{0.9cm} P{1.1cm} P{0.4cm} P{1.1cm} P{0.6cm} P{0.7cm} P{0.7cm} P{0.7cm} P{0.7cm} P{0.7cm} P{1.0cm} P{1.0cm} P{1.2cm} P{1.3cm}} \toprule 
						Array Design & NA & MISC &  NADiS & NSANCS & \textcolor{black}{OCA-SDCA} & \textcolor{black}{OSNA} & Tdis-ULAs & \textcolor{black}{TCNA} & STNA & \textcolor{black}{ZRSA} & GSNA & \textcolor{black}{RSNA} & AULAs & SAULAs &  TSAULAs  & Co-TSAULAs	\\ 
						\midrule
						uDOFs  & \textcolor{black}{95} & \textcolor{black}{113} & \textcolor{black}{155} & \textcolor{black}{165} & \textcolor{black}{117} & \textcolor{black}{169} & \textcolor{black}{173} &
						\textcolor{black}{155} &
						 \textcolor{black}{173} & \textcolor{black}{169} & \textcolor{black}{165} & \textcolor{black}{169} & \textcolor{black}{177} & \textcolor{black}{189} & \textcolor{black}{185} & \textcolor{black}{173} \\
						Holes & \textcolor{black}{60} & \textcolor{black}{24} & \textcolor{black}{120} & 0 & \textcolor{black}{60} & \textcolor{black}{132} & 0  & \textcolor{black}{148} & \textcolor{black}{4} & 0 & \textcolor{black}{122} & \textcolor{black}{124} & 0 & 0 & \textcolor{black}{4} & 0 \\
						CVA  &  94 & 112  & 154 & 164 & 116 & 168 & 172 & \textcolor{black}{154} & 172 & 168 & 164 & \textcolor{black}{168} & \textcolor{black}{176} & \textcolor{black}{188} & \textcolor{black}{184} & \textcolor{black}{172} \\
						S.E (\%) &  57.3 & 57.1 & 54.2 & 100 & 59.1 & 53.8 & 100 & \textcolor{black}{45.3} & 95.5 & 100 & 53.25 & 53.1 & 100 & 100 & 95.8  & 100 \\ 
						\bottomrule
					\end{tabular} \vspace{-1.8em}
			\end{table*}
			
			Fig.~\ref{CoTSAULA} displays the Co-TSAULAs configuration. The associated inter-element spacing set $(\mathbb{A}_{\rm{Co-TSAULAs}})$ of Co-TSAULAs can be formulated as
			\textcolor{black}{\begin{equation} \label{Complimentary}{
						\hspace{-0.9em}\begin{array}{l} \mathbb{A}_\mathrm{Co-TSAULAs} = \{ \underbrace {M,M, \cdots ,M}_{{N_1} - 1},\underbrace {2,2, \cdots ,2}_{{N_2}}, -2{N_1}M \!- \!M \!+ \!3,  \\ 
							\quad \quad \quad \quad \quad \quad 	\quad
							2{N_1}M + M,\underbrace {2,2, \cdots ,2}_{{N_3}-1}, 1\}, 
					\end{array}}
			\end{equation}} 
			where $M, N_{2}$, and $N_{3}$ take the same values as defined in (\ref{set_transformed}), except $N_{1}$ can now be obtained by $N_{1} =  N - M$. 
			
			The location set of Co-TSAULAs $(\mathbb{P}_{\rm{C}})$ is formulated by    
			\begin{equation} \label{Co-TSAULAA}
				\mathbb{P}_{\rm{C}} = {\mathbb{P}_{\rm{c_1}}} \cup {\mathbb{P}_{\rm{c_2}}} \cup {\mathbb{P}_{\rm{c_3}}}, 
			\end{equation}
			where
			\begin{equation} \label{Co-TSAULAA1}
				{\mathbb{P}_{\rm{c_1}}} = \left\{ (M/2) + \left\langle {0,({N_1} - 1)} \right\rangle (M)\right\},
			\end{equation}
			\begin{equation} \label{Co-TSAULAA2}
				\begin{array}{l}
					{\mathbb{P}_{\rm{c_2}}} =  {\mathbb{P}_{\rm{c_{2a}}}} \cup {\mathbb{P}_{\rm{c_{2b}}}}, \\ \quad   \quad \!\! =  \left\{ {({N_1}M - M/2) + \left\langle {1,M/2 - 1} \right\rangle ({\rm{2}})} \right\}  \\
					\quad \quad \quad \cup  \left\{ {(-{N_1}M - M/2 + 1)} \right\},
				\end{array}
			\end{equation}
			\begin{equation} \label{Co-TSAULAA3}
				\begin{array}{l}
					{\mathbb{P}_{\rm{c_3}}} = \left\{({N_1}M + M/2 - 1) + \left\langle {1, M/2 - 1} \right\rangle ({\rm{2}}) \right\} \\
					\quad \quad \quad \cup \left\{ {({N_1}M + 3M/2 - 2)} \right\}.
				\end{array}
			\end{equation}
			\begin{lemma}
				The Co-TSAULAs enjoy the following properties:\\
				(a) The joint sum-difference co-array is continuous in $ \pm  \left\langle {{l_{c_1}}, {l_{c_1}}} \right\rangle$, where ${l_{c_1}} = 2{N_1}M + 3M - 4$. \\
				(b) The DOFs for Co-TSAULA $({D_C})$ can be computed as
				\begin{equation} \label{CTSAULA}
					{D_C} = 4{N_1}M + 6M - 7.
				\end{equation}
			\end{lemma}
			\begin{proof}
				The proof is provided in Appendix D.
			\end{proof} 
			
			\begin{remark}
				Since the Co-TSAULAs originate from the concept of TSAULAs, its $w(2)$ is identical to that of the TSAULAs, while due to a single element of \textit{d} spacing is placed at the end of the third ULA, $w(1) = 1$, and $w(3) = 2$, for $M \ge 6$.
			\end{remark}
			\begin{remark}
				\textcolor{black}{The first ULA of TSAULAs and Co-TSAULAs is identical, except that one element is repositioned. Hence, for a given $N$, Co-TSAULAs provide $M - 3$ in-built physical locations, which is exactly one fewer than the $M - 2$ locations provided by TSAULAs.}  
			\end{remark}

				\section{Numerical Simulations}
				This section conducts simulations to demonstrate the superiority of the AULAs structures by comparing their \textcolor{black}{sum-difference co-array} characteristics, DOFs, mutual coupling leakage, and DOA estimation performance with the existing array designs. We employed the co-array MUSIC algorithm \cite{Coarray, OCA} and used $N = 12$ for all the sparse arrays.
					\begin{figure*} [t]\vspace{-2.25em}\hspace{-1.75em}  
					\begin{minipage}[b]{.4\columnwidth}
						\includegraphics[width=4.95cm, trim=.35cm .35cm .5cm .5cm, height = 5cm]{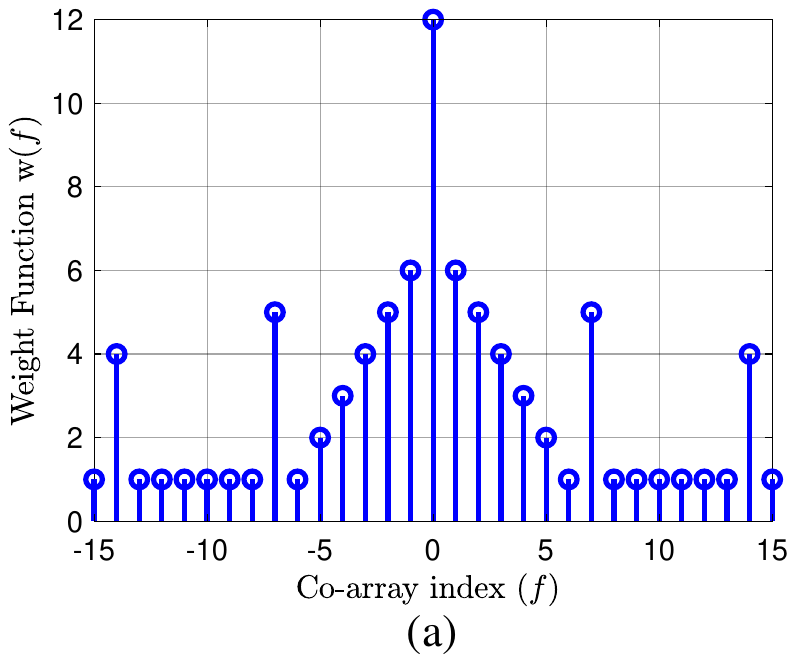}
					\end{minipage}%
					\begin{minipage}[b]{.4\columnwidth}
						\includegraphics[width=4.95cm, trim=.35cm .35cm .5cm .5cm, height = 5cm]{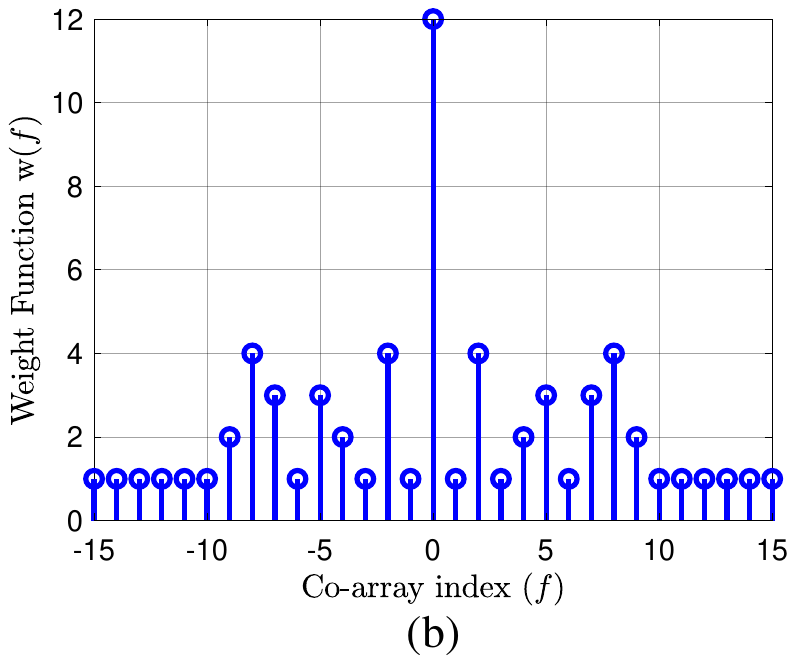}
					\end{minipage}
					\begin{minipage}[b]{.4\columnwidth}
						\includegraphics[width=4.95cm, trim=.35cm .35cm .5cm .5cm, height = 5cm]{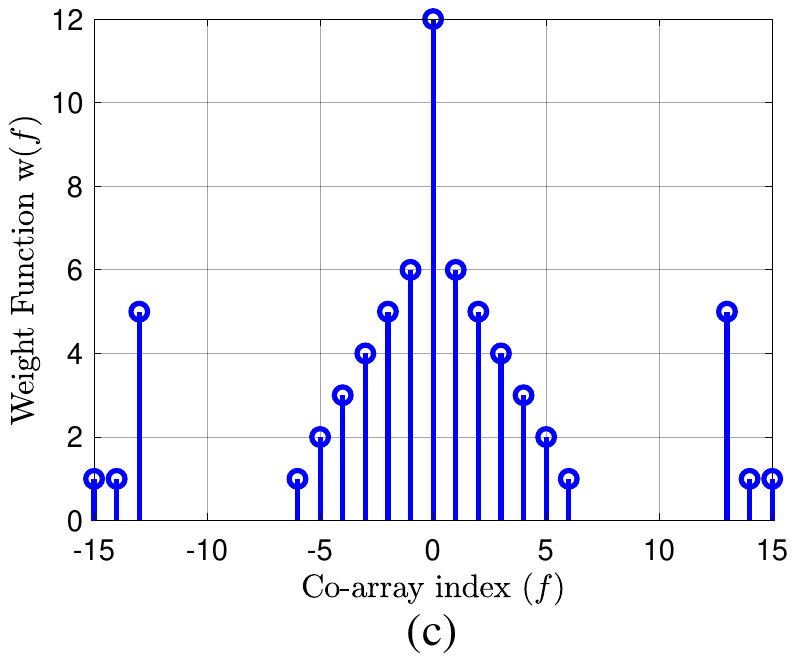}
					\end{minipage}
					\begin{minipage}[b]{.4\columnwidth}
						\includegraphics[width=4.95cm, trim=.35cm .35cm .5cm .5cm, height = 5cm]{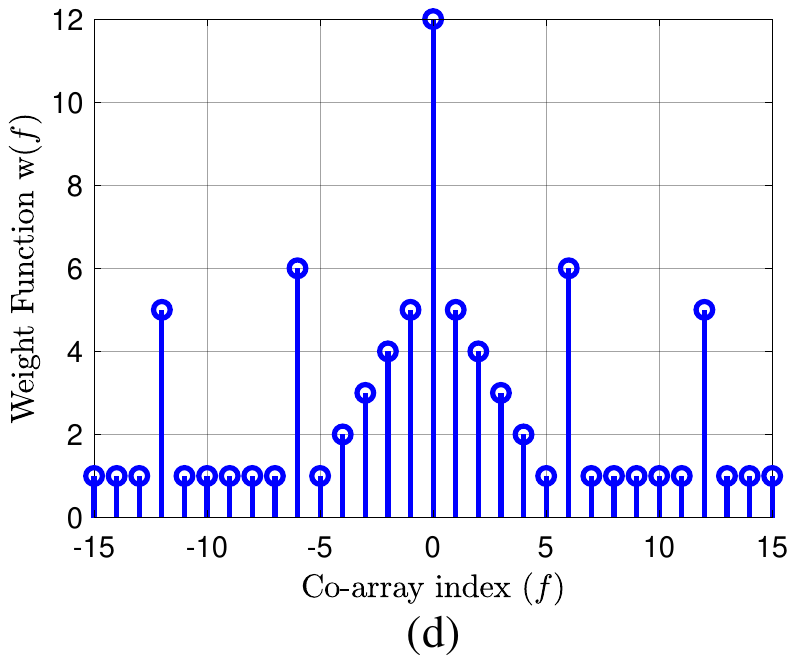}
					\end{minipage}\vspace{-15em} 
					\begin{minipage}[b]{.4\columnwidth}
						\includegraphics[width=4.95cm, trim=.35cm .35cm .5cm .5cm, height = 5cm]{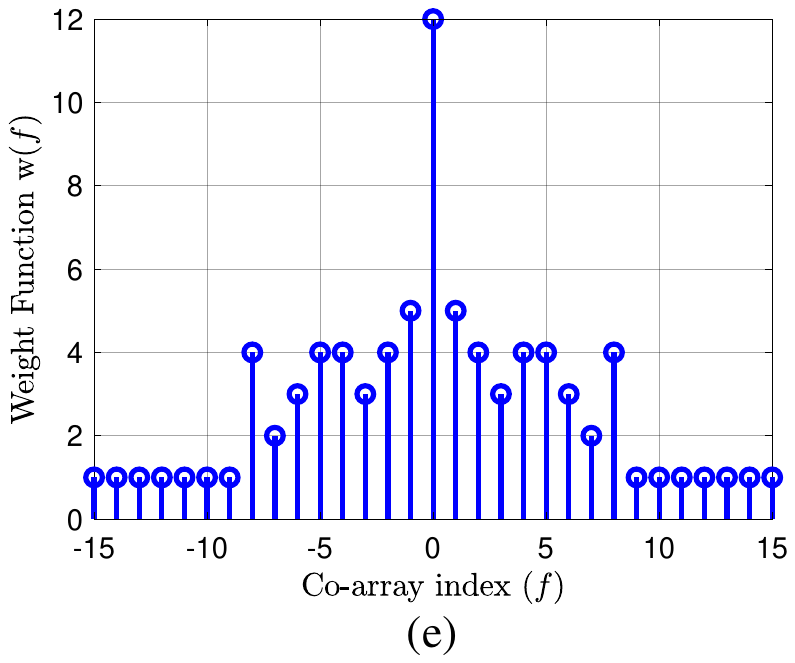}
					\end{minipage} \vspace*{7em}   \\ \hspace*{-1.75em}
					\begin{minipage}[t]{.4\columnwidth}  
						\includegraphics[width=4.95cm, trim=.35cm .35cm .5cm .5cm, height = 5cm]{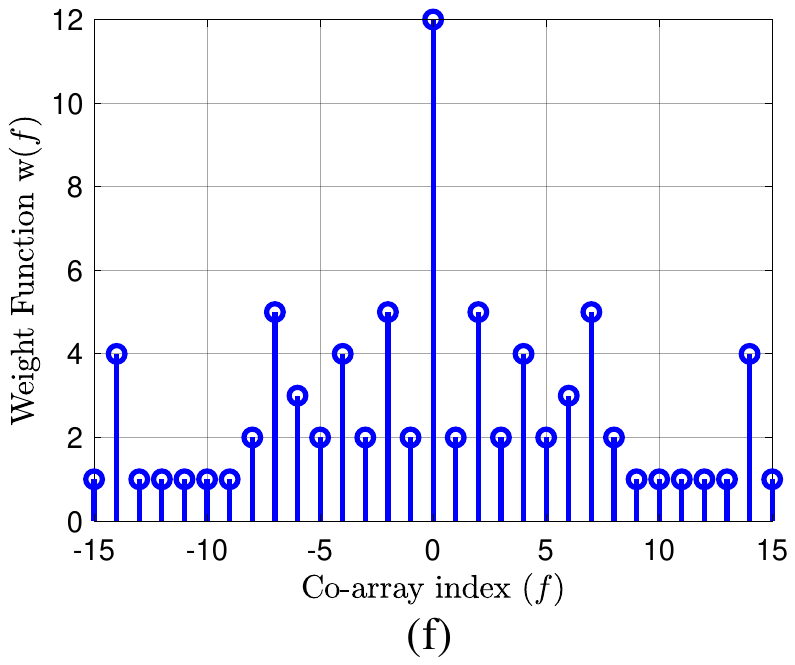}	
					\end{minipage} \hspace*{-0.75em}
					\begin{minipage}[b]{.4\columnwidth}
						\includegraphics[width=4.95cm, trim=.35cm .35cm .5cm .5cm, height = 5cm]{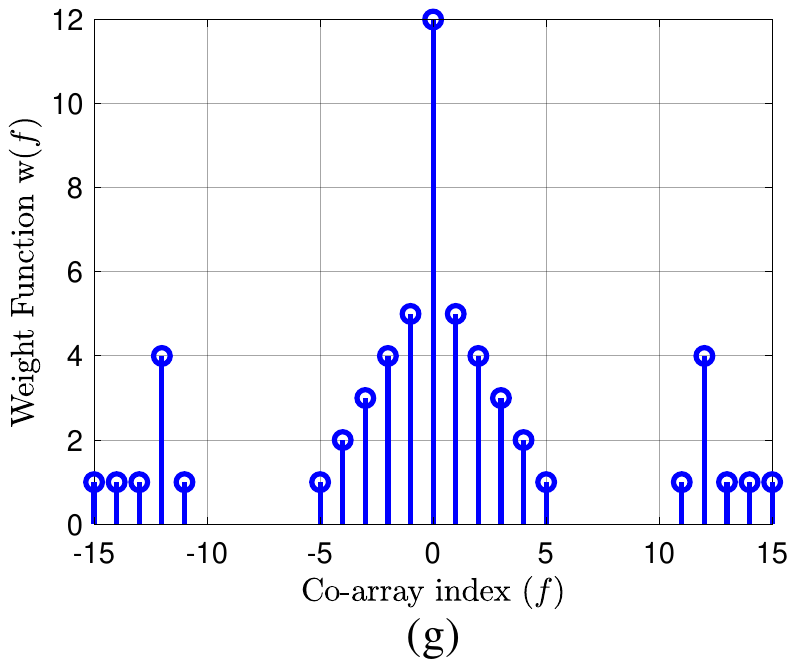}
					\end{minipage} \hspace*{-0.35em}
					\begin{minipage}[b]{.4\columnwidth}
						\includegraphics[width=4.85cm, trim=.35cm .35cm .5cm .5cm, height = 5cm]{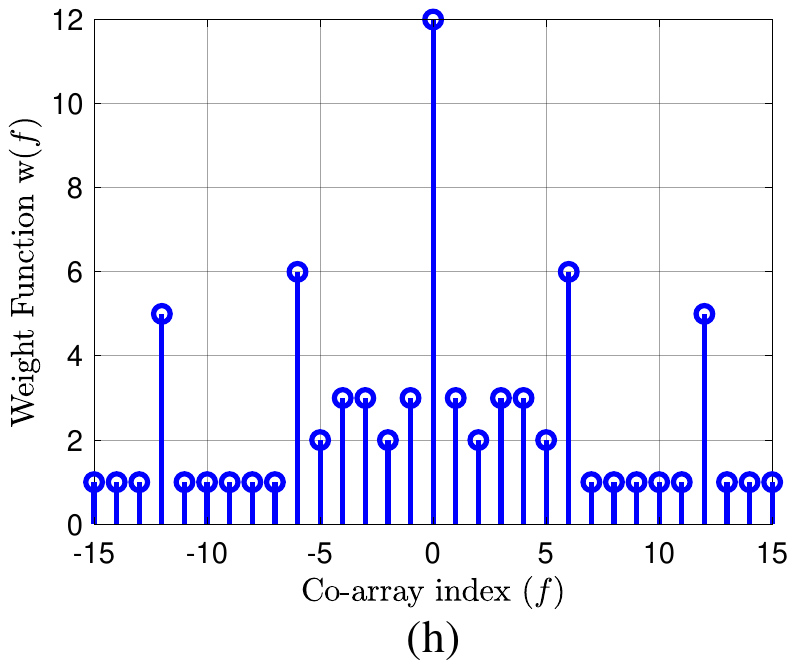}
					\end{minipage} \hspace*{-0.35em}
					\begin{minipage}[b]{.4\columnwidth}
						\includegraphics[width=4.95cm, trim=.35cm .35cm .5cm .5cm, height = 5cm]{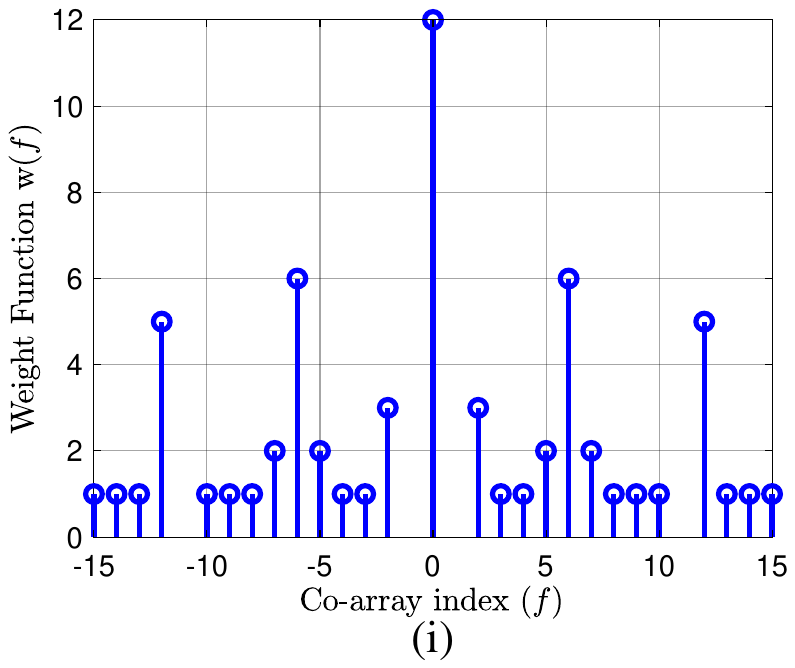}
					\end{minipage} \hspace*{-0.35em}
					\begin{minipage}[b]{.4\columnwidth} 
						\includegraphics[width=4.95cm, trim=.35cm .35cm .5cm .5cm, height = 5cm]{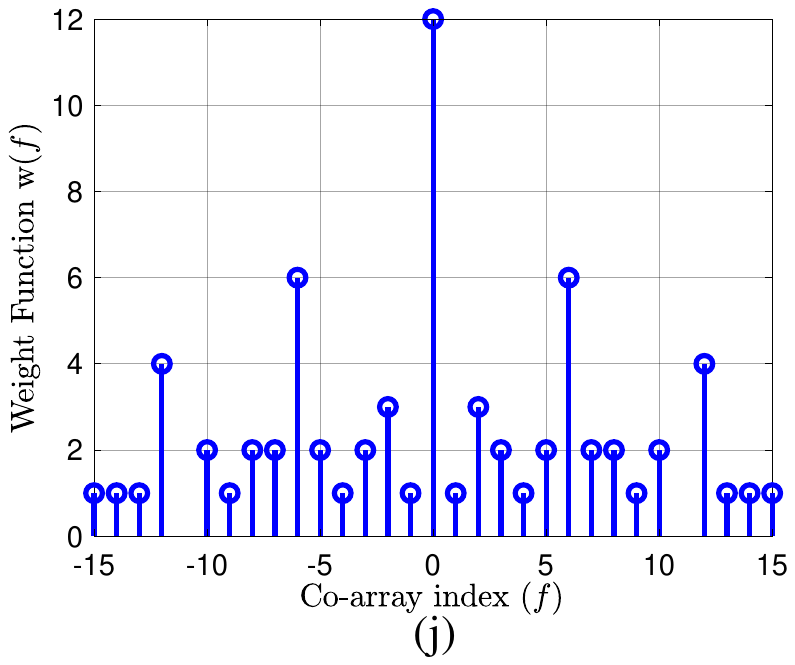}
					\end{minipage} \vspace{-4.95em} \hspace*{-0.7em}
					\caption{The weight functions of different sparse arrays for $N = 12$ elements. (a) NA. (b) MISC (c) NADiS. (d) NSANCS. (e) Tdis-ULAs. (f) STNA. (g) GSNA. (h) SAULAs (i) TSAULAs. (j) Co-TSAULAs. }\label{Weight} 
				\end{figure*}
				\begin{table*} \vspace{-0.75em}  
					\scriptsize
					\caption{SUMMARY Of WEIGHT FUNCTION AND MUTUAL COUPLING LEAKAGE }
					\label{MCT}
					\setlength\tabcolsep{4pt}
						\begin{tabular}{P{1.2cm} P{0.4cm} P{0.7cm} P{0.7cm} P{0.9cm} P{1.1cm} P{0.4cm} P{1.1cm} P{0.6cm} P{0.7cm} P{0.7cm} P{0.7cm} P{0.7cm} P{0.7cm} P{1.0cm} P{1.0cm} P{1.2cm} P{1.3cm}} \toprule 
							Array Design & NA & MISC &  NADiS & NSANCS & \textcolor{black}{OCA-SDCA} & \textcolor{black}{OSNA} & Tdis-ULAs & \textcolor{black}{TCNA} & STNA & \textcolor{black}{ZRSA} & GSNA & \textcolor{black}{RSNA} & AULAs & SAULAs &  TSAULAs  & Co-TSAULAs	\\ 
							\midrule
							$w(1)$ & 6   & 1   & 6   & 5  &  1 & 5 & 5  & \textcolor{black}{1} & 2 & 5   & 5 & 4    &  3   & 3 & 0 & 1  \\ 
							$w(2)$  & 5       & 4        & 5  & 4 & 1  & 4 & 4 & \textcolor{black}{1} & 5   & 4 & 4 & 4   &  2   & 2 & 3  &  3  \\ 
							$w(3)$  & 4       & 1       & 4 & 3  & 6  & 3 & 3 & \textcolor{black}{1} & 2 & 3  & 3 & 3   &  3   & 3 & 1  & 2   \\ 
							$L_C$   & 0.33       & 0.194   & 0.328  & 0.305 & 0.188 & 0.300  & 0.308  & \textcolor{black}{0.161} & 0.241 & 0.307  & 0.299 & 0.278   &  0.249  & 0.249  &  0.140  & 0.189  \\ 
							\bottomrule
						\end{tabular} \vspace{-1.5em}
					\end{table*}
				\subsection{Characteristics Comparison of Different Sparse Arrays}
				\textcolor{black}{} First, we compare the physical and $\mathbb{SDC}$ \textcolor{black}{(non-negative)} arrays of different designs in Fig.~\ref{Charac}, \textcolor{black}{demonstrating how a physical array structure can either enhance the effective virtual aperture or introduce holes that lead to performance loss \cite{Interpolation1, Interpolation2}.} It can be observed from Fig.\,\ref{Charac} that the NA has the least number of continuous virtual lags and suffers from holes at irregular intervals. Likewise, MISC and Pdis-ULAs exhibit similar limitations, as they are also inefficient at exploiting the $\mathbb{SDC}$. In contrast, the designs introduced for the joint optimization of $\mathbb{SC}$ and $\mathbb{DC}$ perform better. As such, the NADiS, OSNA, RSNA, and TCNA exhibit longer hole-free segment in $\mathbb{SDC}$ than the MISC arrays. Nevertheless, their $\mathbb{SDC}$ still congests with numerous holes. By comparison, NSANCS, ZRSA, and Tdis-ULAs produce a hole-free $\mathbb{SDC}$, with Tdis-ULAs generating more virtual lags. The proposed AULAs structures outmatch all the existing designs. Although the TSAULAs design has two holes, it further increases the number of continuous virtual lags. Moreover, the superiority of SAULAs is evident in Fig.\,\ref{Charac}, where they produce a hole-free co-array with highest number of continuous virtual lags.
				
				Next, Table\,\ref{my-label} lists the \textcolor{black}{total} number of uDOFs, S.E, \textcolor{black}{total} number of holes, and CVA \textcolor{black}{in the $\mathbb{SDC}$.}
				The NA has the lowest DOFs capacity and the smallest CVA. Moreover, since the holes appear in the central part of its co-array, the S.E is also reduced. Notably, OCA-SDCA contains the same number of holes as the NA but delivers higher uDOFs.  Although OSNA contains more holes than the OCA-SDCA and NADiS, resulting in a lower S.E, it produces comparatively more uDOFs. \textcolor{black}{Although TCNA produces the same number of uDOF as the NADiS, it delivers the smallest S.E, implying significant waste of aperture resources.} Meanwhile, RSNA achieve the same number of uDOFs as OSNA. The NSANCS, ZRSA, and Tdis-ULAs obtain maximum S.E, with the latter offering more uDOFs. Despite the TSAULAs containing four holes, they generate higher uDOFs and a larger CVA than the Tdis-ULAs, while also achieving superior S.E compared with RSNA, OSNA, and OCA-SDCA. The SAULAs exhibit maximum S.E and highest uDOFs with largest CVA among all the designs.
								\subsection{Weight Functions and Mutual Coupling Leakage}
					The following section evaluates the weight functions and MC leakage of different sparse arrays. We characterize the MC model by $c_{1} = e^{j \pi/3}$, with $c_{\textcolor{black}{q}} = c_{1}e^{(-j(\textcolor{black}{q}-1)\pi/8)}/\textcolor{black}{q}$ for $2 \le \textcolor{black}{q} \le B,$ where $B$ = 100. Fig.\,\ref{Weight} displays the weight functions, where the NA and NADiS arrays are observed to own the highest values for the first three weight functions $(w(1) = 6, w(2) = 5, w(3) = 4)$, indicating a prominent presence of dense ULA in their physical structures. Although the values of NSANCS, Tdis-ULAs, and GSNA are identical for $(w(1) = 5, w(2) = 4, w(3) = 3)$, they are smaller than those of NADiS. AULAs and SAULAs share identical weight function values $(w(1) = 3, w(2) = 2, w(3) = 3)$. 
					The proposed TSAULAs design $(w(1) = 0, w(2) = 3, w(3) = 1)$ offers excellent sparsity among all other designs, as evidenced  by its smallest weight function values. The Co-TSAULAs inherit this sparsity property of TSAULAs and renders smaller values of $(w(1) = 1, w(2) = 3, w(3) = 2)$ than the Tdis-ULA, SAULAs, and STNA $(w(1) = 2, w(2) = 5, w(3) = 2)$. 
					
					Moreover, a summary of weight functions and coupling leakage is provided in Table\,\ref{MCT}. We observe that the NA and NADiS exhibit the highest value of ${L_C}$, indicating that their estimation performance is significantly affected by the MC. The STNA, SAULAs, and AULAs induce weaker MC than NSANCS, Tdis-ULAs, and GSNA. Although the OCA-SDCA exhibits a relatively small ${L_C}$ value, it was shown earlier that the OCA-SDCA yields significantly lower uDOFs and produces numerous holes. Notably, the proposed TSAULAs design owns the smallest value of ${L_C}$ in \textcolor{black}{Table\,II}, implying that it is the least sensitive to MC. Meanwhile, Co-TSAULAs yield a moderate ${L_C}$ value, lower than those of SAULAs and STNA and comparable to that of OCA-SDCA. These results reflect that the existing designs are constrained by a three-way performance trade-off, whereas the proposed AULAs framework resolves this limitation through its distinctive advantages.
					
				\textcolor{black}{\subsection{uDOFs and Mutual Coupling Leakage }
						The uDOFs are a critical quantitative metric used to characterize performance based on the number of sources a sparse array can detect. Fig.~\ref{DOFVsSen} compares the uDOFs capacity of different sparse arrays against an increasing number of elements. Since the NA are designed to primarily optimize the $\mathbb{DC}$, it produces the smallest uDOFs in Fig.~\ref{DOFVsSen}(a), confirming the inefficient utilization of the $\mathbb{SDC}$. Though the OCA-SDCA is a coprime-based structure, it achieves more uDOFs than the NA due to more effective utilization of the $\mathbb{SDC}$. While the uDOFs of most existing designs are comparable as they display almost similar curves, the proposed designs, especially SAULAs and TSAULAs, produce higher uDOFs and convincingly outperform all the existing designs. For instance, the TSAULAs with only $N = 19$ elements offer the same DOFs as the OCA-SDCA with $N = 25$ elements. }
					\begin{figure} \vspace{-0.5em} 
						\hspace{0.5em} 	\includegraphics[height=3.5cm,width=8.5cm]{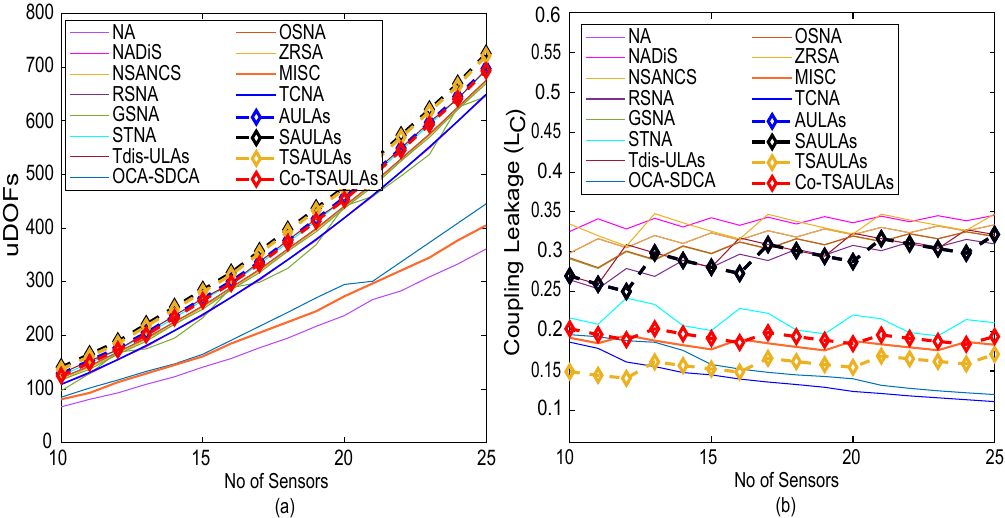} \vspace{-1em} 
						\caption{\textcolor{black}{uDOFs and MC leakage performances.}} \vspace{-1.5em}
						\label{DOFVsSen}
					\end{figure} 
					
					\textcolor{black}{Next, we simulated the MC leakage performance against an increasing number of elements. Fig.~\ref{DOFVsSen}(b) displays the coupling leakage values of various designs. As observed, most designs based on the nested topology suffer from strong MC because they contain dense sections with \textit{d}-spaced elements. STNA reduces these spacings by splitting and transforming the dense ULA of NA, thereby inducing a weak MC. Although OCA-SDCA has the smallest ${L_C}$ value for $N \ge 17$, it outperforms only NA in terms of uDOFs. Hence, the OCA-SDCA, like other existing design, fails to exhibit a balanced performance between the two metrics. In contrast, TSAULAs and Co-TSAULAs have a striking balance between the uDOFs and MC, especially the TSAULAs. When $N$ is smaller, TSAULAs experience least MC effects while simultaneously providing the second-highest uDOFs. Although TSAULAs design ranks third only to TCNA and OCA-SDCA in terms of MC for higher $N$, its uDOFs are significantly larger than both designs while maintaining higher S.E in contrast to OCA-SDCA and TCNA.}

					\vspace{-1em}\subsection{DOA Spectrum in the Absence of Mutual Coupling}
					\textcolor{black}{This simulation compares the estimation performance without MC.} We assume $Z = 55$ sources, located at $ {\theta _z} =  -49^\circ + 110^\circ\left( {z - 1} \right)/54,$ $z \in \left\langle {1, Z} \right\rangle$, impinging on the arrays with an equal power. Fig.\,\ref{MSA} plots the MUSIC spectrum of different sparse arrays alongside the root mean square error (RMSE), where $T = 600$ and $SNR = 0$ dB. \textcolor{black}{We observed that the MUSIC spectra of OCA-SDCA are of low quality and fail to generate peaks for all the incoming source signals, which can be expected because the uDOFs of OCA-SDCA are nearly equal to the number of sources. Hence, the OCA-SDCA renders the largest RMSE value. Although the NADiS performs better than the OCA-SDCA, it still lacks accuracy in resolving the sources, thus producing the second-largest RMSE value.} The SAULAs scheme detects and resolves the source DOAs most accurately among the other sparse arrays and provides a high-quality spectrum, with the TSAULAs demonstrating the second-best performance. In particular, SAULAs improve the detection quality by approximately 98$\%$ and 33$\%$ compared to the OCA-SDCA and ZRSA, respectively.
					
					\begin{figure} [btp] \vspace{-0.2em} 
						\begin{minipage}[b]{1\columnwidth} 
							\centering
							\includegraphics[width=8cm, trim=.35cm .35cm .5cm .5cm, height = 6.5cm]{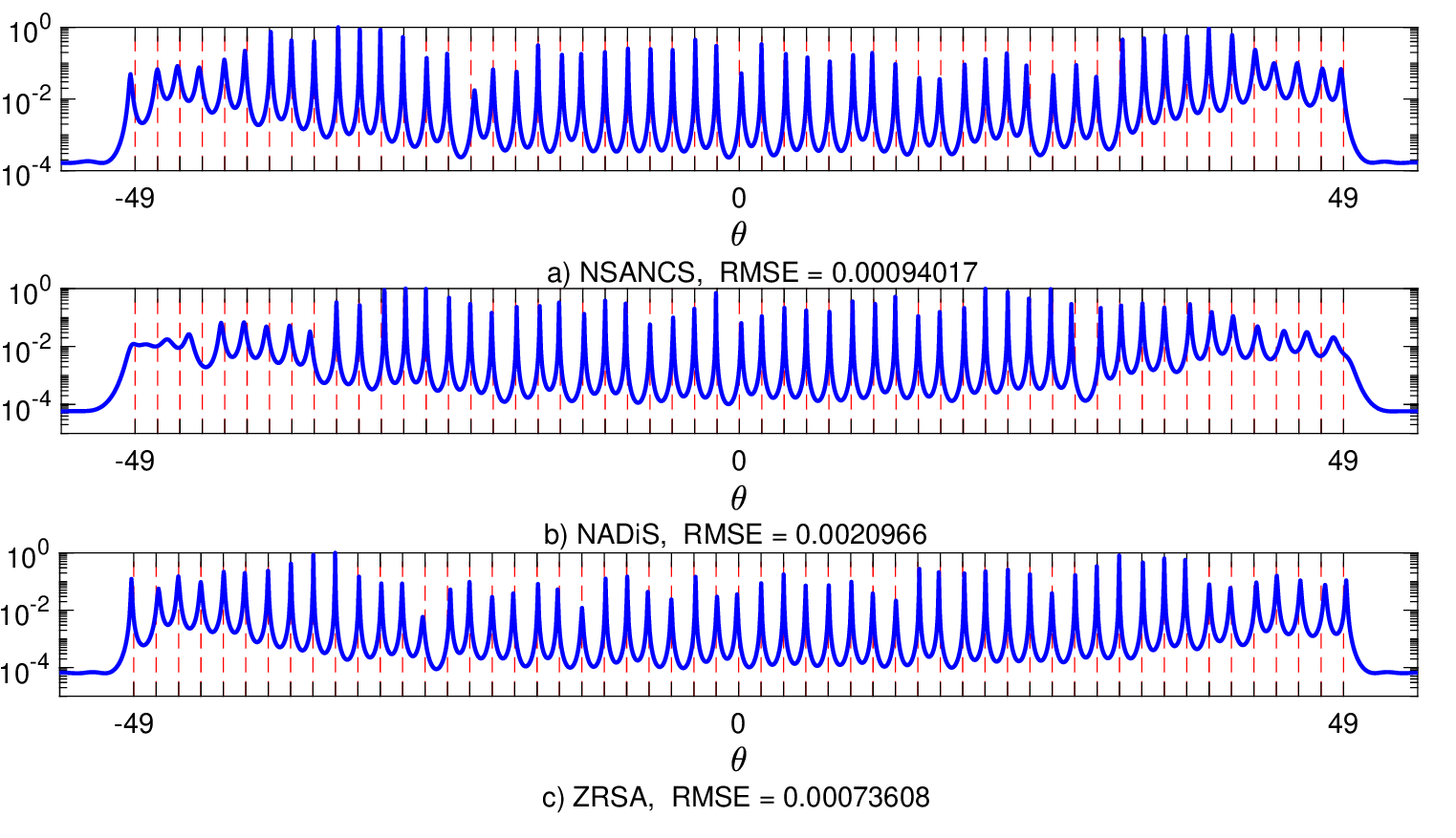}  \vspace{0.5em} \hspace{-0.5em}
						\end{minipage} 
						\begin{minipage}[b]{1\columnwidth}
							\centering
							\includegraphics[width=8cm, trim=.35cm .35cm .5cm .5cm, height = 6.5cm]{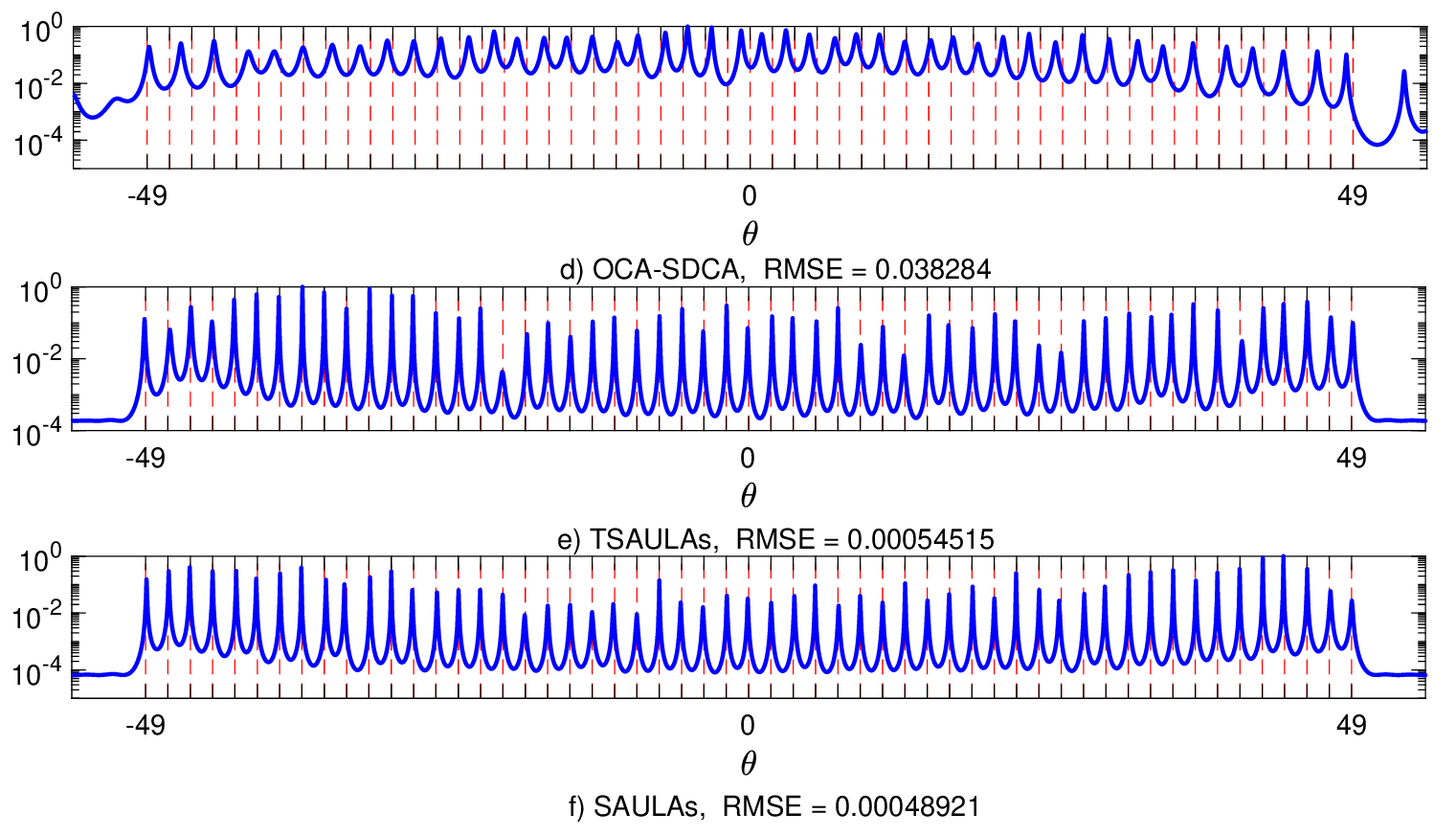}
						\end{minipage} \hfill{} \vspace{-1.75em}
						\caption{MUSIC spectra of NSANCS, NADiS, ZRSA, OCA-SDCA, TSAULAs, and SAULAs in the absence of MC, when $N = 12$ and  sources are located at $ {\theta _z} =  -49^\circ + 98^\circ\left( {z - 1} \right)/54,$ $z \in \left\langle {1, 55} \right\rangle.$ }\label{MSA}
					\end{figure}
					\begin{figure}  [btp] \vspace{-1em} 
						\begin{minipage}[b]{1\columnwidth} 
							\centering
							\includegraphics[width=8cm, trim=.35cm .35cm .5cm .5cm, height = 6.5cm]{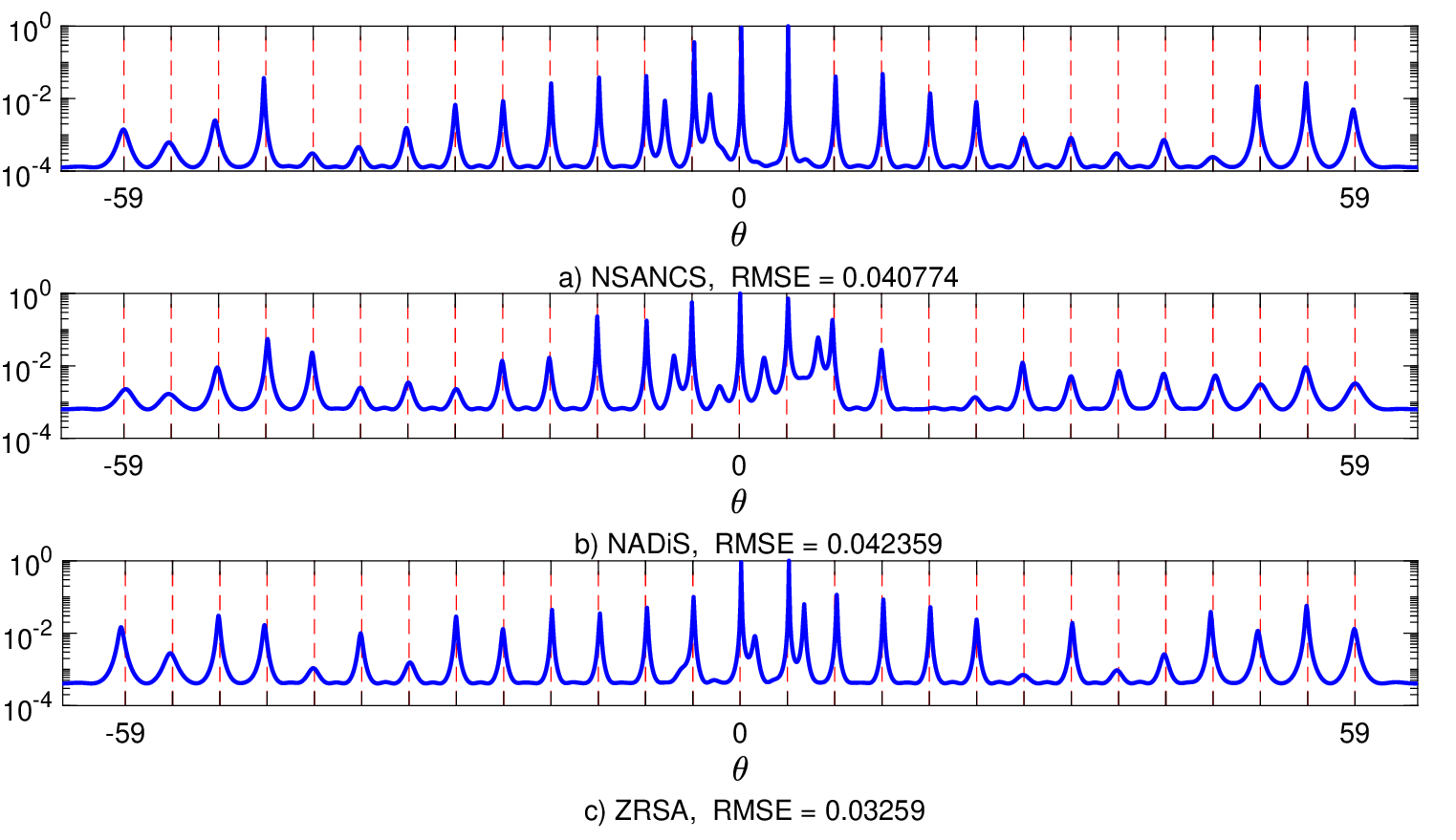} \vspace{0.7em}\hspace{-0.5em}
						\end{minipage} 
						\begin{minipage}[b]{1\columnwidth}
							\centering
							\includegraphics[width=8cm, trim=.35cm .35cm .5cm .5cm, height = 6.5cm]{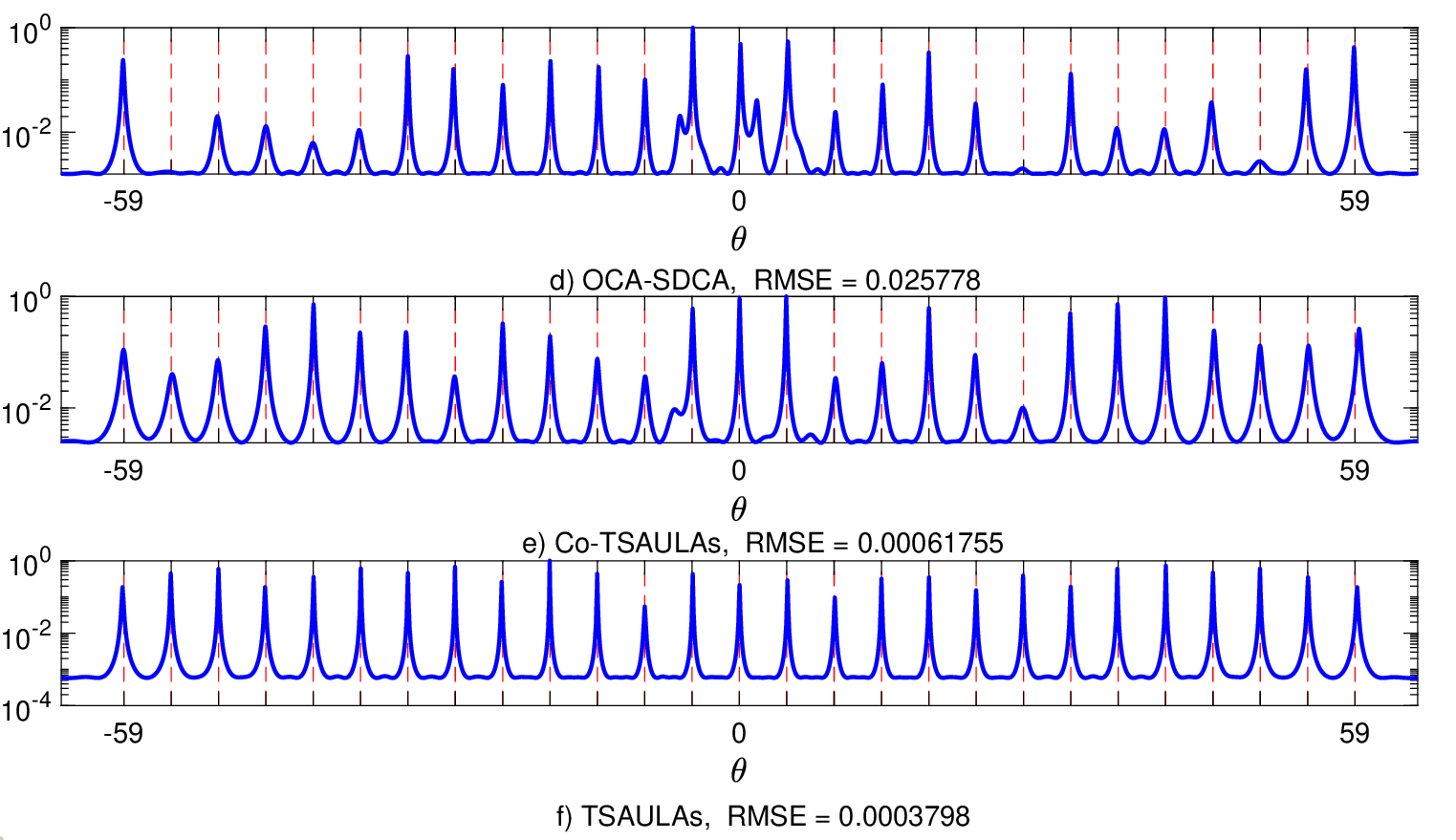}
						\end{minipage} \vspace{-1.75em}
						\caption{MUSIC spectra of NSANCS, NADiS, ZRSA, OCA-SDCA, Co-TSAULAs, and TSAULAs in the presence of MC, when $N = 12$ and sources are located at $ {\theta _z} =  -59^\circ + 118^\circ\left( {z - 1} \right)/26,$ $z \in \left\langle {1, 27} \right\rangle.$}\label{MCP}
					\end{figure}

					\subsection{DOA Spectrum Under Mutual Coupling}
					This simulation analyzes the estimation performance in the presence of strong MC, \textcolor{black}{using the coupling model defined in Section V-B}. We assume $Z = 27$ sources, located at $ {\theta _z} =  -59^\circ + 118^\circ\left( {z - 1} \right)/26,$ $z \in \left\langle {1, Z} \right\rangle$, impinging on sparse arrays with equal power.  The SNR is fixed to 0 dB and $T = 1000$ snapshots are used to estimate the source DOAs in Fig.\,\ref{MCP}. Since the Co-TSAULAs design is more robust against MC than the SAULAs, we used it for the performance comparison. As observed, the NADiS performs poorly, considering that the sources are fewer now. This is because NADiS is comparatively more sensitive to MC, as witnessed in Table\,II. \textcolor{black}{The performance of NSANCS and ZRSA also deteriorated sharply, drawing an adverse MC impact from their dense ULAs. Although the OCA-SDCA performs better than the NSANCS and ZRSA}, its peaks are not as aligned with the true source directions as those of the Co-TSAULAs. Meanwhile, the TSAULAs outperform all the designs by producing a high-quality MUSIC spectrum and detecting the sources most accurately. This confirms the excellent benefits of the unified AULAs family, where TSAULAs perform best in the presence of MC, while SAULAs stand out in the absence of MC.  
					\section{Conclusion}
					This paper introduces a new array design framework, termed AULAs, for DOA estimation of NCS. The AULAs configure three ULAs and two separate elements. Through intelligent relocation of elements, three new variants of AULAs are synthesized, each exhibiting unique features. The AULAs family exploits the joint sum-difference co-array more effectively than the existing design schemes. Closed-form expressions are provided for physical element placement, co-array analysis, and weight functions. Compared with existing designs for NCS, the AULAs structures demonstrate superior DOA estimation performance, higher uDOFs, a larger virtual aperture, and reduced mutual coupling, all without requiring additional elements. \textcolor{black}{Furthermore, the proposed AULAs family operates within the well-established non‑circular signal framework, yielding physically realizable arrays with significant performance benefits, as demonstrated throughout the paper.} The results of extensive simulations also confirm these conclusions.  
							\section*{ Appendix A}
							\textcolor{black}{Throughout this appendix, we use the array design equations of the AULAs as defined in (\ref{PV}), (\ref{PV1}),  (\ref{AULAA})-(\ref{AULAA3})}: 
							
							The positive difference set between ${\mathbb{P}_{\rm{a_2}}}$ and ${\mathbb{P}_{\rm{a_1}}}$ is ${\mathbb{P}_{\rm{a_2}}} - {\mathbb{P}_{\rm{a_1}}} = \left\langle {0, {N_1} - 1} \right\rangle (M) + \left\langle {M/2, M - 1} \right\rangle$, which is continuous in range $\left\langle {0, {N_1}M - 1} \right\rangle$ except for holes at ${\mathbb{DH}} = {\mathbb{DH}_{\rm{1}}} \cup {\mathbb{DH}_{\rm{2}}} \cup {\mathbb{DH}_{\rm{3}}} =  \left\langle {0, {N_1} - 1} \right\rangle (M) + \left\langle {0, M/2 - 1} \right\rangle$, where ${\mathbb{DH}_{\rm{1}}} = \left\langle {0, {N_1} - 1} \right\rangle (M)$,  ${\mathbb{DH}_{\rm{2}}} = \left\langle {1, {N_1} - 1} \right\rangle (M) + \left\langle {1, M/2 - 1} \right\rangle$ and  ${\mathbb{DH}_{\rm{3}}} = \left\langle {1,  M/2 - 1} \right\rangle$. It is evident that the holes, ${\mathbb{DH}_{\rm{1}}}$, can be filled by ${\mathbb{P}_{\rm{a_1}}} - {\mathbb{P}_{\rm{a_1}}} = \left\langle {0, {N_1} - 1} \right\rangle (M).$ Likewise, ${\mathbb{DH}_{\rm{2}}}$ can be filled by ${\mathbb{P}_{\rm{a_3}}} - {\mathbb{P}_{\rm{a_1}}} = \left\langle {1, {N_1}} \right\rangle (M) + \left\langle {1, M/2 - 1} \right\rangle$ and ${\mathbb{DH}_{\rm{3}}}$ can be filled by ${\mathbb{P}_{\rm{a_2}}} - {\mathbb{P}_{\rm{a_2}}} = \left\langle {0, M/2 - 1} \right\rangle$. Therefore, the entire set of holes ${\mathbb{DH}}$ is filled and we establish the virtual lags are continuous till  ${N_1}M - 1.$  
							We also noticed that the difference set ${\mathbb{P}_{\rm{a_3}}} - {\mathbb{P}_{\rm{a_1}}}$ provide continuous virtual lags, located in  ${N_1}M + \left\langle {1, M/2 - 1} \right\rangle$. 
							Analyzing the difference sets, ${\mathbb{P}_{\rm{a_3}}} - {\mathbb{P}_{\rm{a_1}}}$ and ${\mathbb{P}_{\rm{a_2}}} - {\mathbb{P}_{\rm{a_1}}}$, it can be concluded that the positive difference co-array of AULAs generates continuous lags up to $ {l_{a_1}} =  {N_1}M + M/2 - 1$, except for a single hole at ${N_1}M$. We will prove that this hole is filled by the sum co-array.
							
							In the second part, we analyze the virtual locations generated by the sum co-array. The positive sum set between ${\mathbb{P}_{\rm{a_1}}}$ and ${\mathbb{P}_{\rm{a_2}}}$ is ${\mathbb{P}_{\rm{a_1}}} + {\mathbb{P}_{\rm{a_2}}} = \left\langle {{N_1} - 1, 2{N_1} - 2} \right\rangle (M) + \left\langle {M/2, M - 1} \right\rangle$, which has continuous virtual lags in the range $\left\langle {{N_1}M - M/2, 2{N_1}M - M - 1} \right\rangle$, except for holes at ${\mathbb{SH}}_{\rm{L_1}} = {\mathbb{SH}}_{\rm{l_{11}}} \cup {\mathbb{SH}}_{\rm{l_{12}}} = \left\langle {{N_1}, 2{N_1} - 2} \right\rangle (M) + \left\langle {0, M/2 - 1} \right\rangle$, where ${\mathbb{SH}}_{\rm{l_{11}}} = \left\langle {{N_1}, 2{N_1} - 2} \right\rangle (M).$ ${\mathbb{P}_{\rm{a_1}}} + {\mathbb{P}_{\rm{a_1}}} = \left\langle {0, 2{N_1} - 2}\right\rangle (M)$ fill the set of holes ${\mathbb{SH}}_{\rm{l_{11}}}.$ 
							Likewise, ${\mathbb{SH}}_{\rm{l_{12}}} = \left\langle {{N_1}, 2{N_1} - 2} \right\rangle (M) + \left\langle {1, M/2 - 1} \right\rangle$ can be filled by ${\mathbb{P}_{\rm{a_1}}} + {\mathbb{P}_{\rm{a_3}}} = \left\langle {{N_1}, 2{N_1} - 1} \right\rangle (M) + \left\langle {1, M/2 - 1} \right\rangle.$ 
							
							Moreover, ${\mathbb{P}_{\rm{a_2}}} + {\mathbb{P}_{\rm{a_2}}}$ is continuous in range $(2{N_1}M - M) + \left\langle {0, M - 2} \right\rangle$. While ${\mathbb{P}_{\rm{a_2}}} + {\mathbb{P}_{\rm{a_3}}}$ maintains the continuity in the range $ (2{N_1}M  - M) + \left\langle {M/2 + 1, 3M/2 - 2} \right\rangle$, extending the continuous ULA segment up to $2{N_1}M + M/2 - 2$. Finally, the sum set between ${\mathbb{P}_{\rm{a_3}}} + {\mathbb{P}_{\rm{a_3}}}$ is continuous in the range of $(2{N_1}M - M) + \left\langle {M + 2, 2M - 2} \right\rangle$ and further increases the range of continuous virtual part to $2{N_1}M + M - 2$. Hence, the sum co-array is found to be continuous in  $\left\langle {{l_{a_2}}, {l_{a_3}}} \right\rangle$, where 
							\begin{equation}
								\left\{ {\begin{array}{*{20}{l}}
										{l_{a_2}} = {N_1}M - M/2,\\
										{l_{{a_3}}} = {2{N_1}M + M - 2.}
								\end{array}} \right.
							\end{equation} 
							As can be noticed,  ${\mathbb{P}_{\rm{a_1}}} + {\mathbb{P}_{\rm{a_1}}}$ yields a virtual lag at ${N_1}M$ that overlaps exactly with the hole in the difference co‑array. Consequently, the joint sum-difference co-array of AULAs is proved to be continuous in $\left\langle {0, {2{N_1}M + M - 2}} \right\rangle$, and the total DOFs are obtained by $4{N_1}M + 2M - 3$.
							\section*{Appendix B}
							\textcolor{black}{Since \textbf{Lemma 2} covers the co-array properties of SAULAs, this proof relies on the definitions provided in equations (\ref{PV}), (\ref{PV1}), (\ref{SAULAA})-(\ref{SAULAA3}), which will be used throughout this appendix}: 
							
							Lets first examine the virtual locations that the positive  difference sets of ${\mathbb{P}_{\rm{S}}}$ can obtain. Recall that the SAULAs is a shifted version of AULAs; therefore, their associated inter-element spacing sets are identical. Meanwhile, the difference between their last and first physical elements remains the same. Thus, it is easy to conclude that the difference co-array properties of the AULAs also stands true for SAULAs. Such as, the positive difference sets of SAULAs also possess continuous virtual lags $ {l_{s_1}} = \left\langle {0, N_{1}M + M/2 - 1} \right\rangle$, except a single hole occurring at $N_{1}M$ (see Fig. 4). We will analyze that this hole align with the virtual lag in the sum co-array.
							
							Next, we discuss the virtual lags provided by the sum co-array of SAULAs. The positive sum set between ${\mathbb{P}_{\rm{s_1}}}$ and ${\mathbb{P}_{\rm{s_2}}}$ is ${\mathbb{P}_{\rm{s_1}}} + {\mathbb{P}_{\rm{s_2}}} = \left\langle {{N_1}, 2{N_1} - 1} \right\rangle (M) + \left\langle {M/2, M - 1} \right\rangle$, which generates continuous lags in $\left\langle {{N_1}M + M/2, 2{N_1}M - 1} \right\rangle$, except for holes at ${\mathbb{SH}}_{\rm{L_1}} = {\mathbb{SH}}_{\rm{l_{11}}} \cup {\mathbb{SH}}_{\rm{l_{12}}} = \left\langle {{N_1} + 1,2{N_1} - 1} \right\rangle (M) + \left\langle {0, M/2 - 1} \right\rangle$, where ${\mathbb{SH}}_{\rm{l_{11}}} = \left\langle {{N_1} + 1, 2{N_1} - 1} \right\rangle (M).$ Clearly, the sum set ${\mathbb{P}_{\rm{s_1}}} + {\mathbb{P}_{\rm{s_1}}} = \left\langle {1, 2{N_1} - 1}\right\rangle (M)$ can fill the set of holes ${\mathbb{SH}}_{\rm{l_{11}}}$. Also, it generates virtual lag at ${N_1}M$, which splices with the hole in the difference co-array. Likewise, ${\mathbb{SH}}_{\rm{l_{12}}} = \left\langle {{N_1} + 1, 2{N_1} - 1} \right\rangle (M) + \left\langle {1, M/2 - 1} \right\rangle$ can be filled by ${\mathbb{P}_{\rm{s_1}}} + {\mathbb{P}_{\rm{s_3}}} = \left\langle {{N_1} + 1, 2{N_1}} \right\rangle (M) + \left\langle {1, M/2 - 1} \right\rangle.$ Besides, ${\mathbb{P}_{\rm{s_2}}} + {\mathbb{P}_{\rm{s_2}}}$ is continuous in range $(2{N_1}M - M) + \left\langle {M, 2M - 2} \right\rangle$ and ${\mathbb{P}_{\rm{s_2}}} + {\mathbb{P}_{\rm{s_3}}}$ produces continuous lags in $ 2{N_1}M + \left\langle {M/2 + 1, 3M/2 - 2} \right\rangle$. Hence, the virtual lags are continuous up to $2{N_1}M + 3M/2 - 2$. Meanwhile, the sum set ${\mathbb{P}_{\rm{s_3}}} + {\mathbb{P}_{\rm{s_3}}}$ has hole-free lags in $ \left\langle {2{N_1}M + M + 2, 2{N_1}M + 2M - 2} \right\rangle$, extending the uniform range up to $2{N_1}M + 2M - 2$. 
							
							As a result, the sum co-array of SAULAs is continuous in $\left\langle {{l_{s_2}}, {l_{s_3}}} \right\rangle$, where ${l_{s_2}} = {N_1}M + M/2$ and 
							$	{l_{s_3}} =	2{N_1}M + 2M - 2$. Meanwhile, the hole in  difference co-array is also filled by the sum co-array. Since the sum (with its negative) co-array and difference co-array develop  symmetric virtual locations, the joint sum-difference co-array of SAULAs is proved to be continuous in $\left\langle {{-2{N_1}M - 2M + 2}, {2{N_1}M + 2M - 2}} \right\rangle$, and thus, the total DOFs are obtained by $4{N_1}M + 4M - 3$.

							\section*{Appendix C}
							\textcolor{black}{As \textbf{Lemma 3} deals with the co-array properties of TSAULAs, in this section, we use the array design equations of TSAULAs provided in (\ref{set_transformed}), (\ref{TSAULAA})-(\ref{TSAULAA3})}:
							
							The positive difference set between ${\mathbb{P}_{\rm{t_{2a}}}}$ and ${\mathbb{P}_{\rm{t_1}}}$ is ${\mathbb{P}_{\rm{t_{2a}}}} - {\mathbb{P}_{\rm{t_1}}} = \left\{ {\left\langle {0, {N_1} - 1} \right\rangle (M) + \left\langle {1, M/2 - 1} \right\rangle (2)} \right\}$, which is continuous in $\left\langle {0, {N_1}M - 2} \right\rangle$, except for holes at ${\mathbb{DH}} = {\mathbb{DH}_{\rm{1}}} \cup {\mathbb{DH}_{\rm{2}}} \cup  {\mathbb{DH}_{\rm{3}}} \cup {\mathbb{DH}_{\rm{4}}}$ , where ${\mathbb{DH}_{\rm{1}}} = \left\langle {0,{N_1} - 1} \right\rangle (M)$, ${\mathbb{DH}_{\rm{2}}} = 1 + \left\{ {\left\langle {1,{N_1}} \right\rangle (M) + \left\langle {0,M/2 - 2} \right\rangle (2)} \right\}$, ${\mathbb{DH}_{\rm{3}}} = {(M - 1) + \left\langle {0, {N_1}-2} \right\rangle (M)}$. It is evident that the holes ${\mathbb{DH}_{\rm{1}}}$ can be filled by ${\mathbb{P}_{\rm{t_1}}} - {\mathbb{P}_{\rm{t_1}}} = \left\langle {0,{N_1} - 1} \right\rangle (M).$ Likewise, ${\mathbb{DH}_{\rm{2}}}$ can be filled by ${\mathbb{P}_{\rm{t_3}}} - {\mathbb{P}_{\rm{t_1}}} = 1 + \{\left\langle {1,{N_1}} \right\rangle (M) + \left\langle {0, M/2 - 2} \right\rangle (2)\}$ and ${\mathbb{DH}_{\rm{3}}}$ can be filled by ${\mathbb{P}_{\rm{t_{1}}}} + {\mathbb{P}_{\rm{t_{2b}}}} =  - \{(M - 1) + \left\langle {0, {N_1}-2} \right\rangle (M)\}$. While the hole on ${\mathbb{DH}_{\rm{4}}} = {1}$ can be filled by  ${\mathbb{P}_{\rm{t_{2a}}}} + {\mathbb{P}_{\rm{t_{2b}}}}$, which generates the virtual lags at $\left\langle -1, (M/2 - 2)(-2) \right\rangle$. Since the sum (with its negative) co-array develops a symmetric virtual locations, it can be inferred that the virtual lags are continuous up to $\left\langle {0, {N_1}M - 2} \right\rangle.$

							Next, we consider the sum set between ${\mathbb{P}_{\rm{t_1}}}$ and ${\mathbb{P}_{\rm{t_{2a}}}}$ for the TSAULAs. ${\mathbb{P}_{\rm{t_1}}} + {\mathbb{P}_{\rm{t_{2a}}}} = {\left\langle {1,{N_1}} \right\rangle (M) + \left\langle {{N_1}M/2 - M/2 + 1,{N_1}M/2 - 1} \right\rangle (2)} $ is continuous in $ \left\langle {N_1}M + 2, 2{N_1}M - 2 \right\rangle$, except for holes at ${\mathbb{SH}}_{\rm{L_{1}}} = {\mathbb{SH}}_{\rm{l_{11}}} \cup {\mathbb{SH}}_{\rm{l_{12}}} \cup {\mathbb{SH}}_{\rm{l_{13}}}$, where ${\mathbb{SH}}_{\rm{l_{11}}} = \left\langle {{N_1} + 1, 2{N_1} - 1} \right\rangle (M), {\mathbb{SH}}_{\rm{l_{12}}}  = \left( {{N_1}M + 1} \right) + \left\{ {\left\langle {1:{N_1}-1} \right\rangle \left( M \right) + \left\langle {0:M/2 - 2} \right\rangle \left( 2 \right)} \right\}$, ${\mathbb{SH}}_{\rm{l_{13}}} = \left\langle {0:{N_1} - 1} \right\rangle \left( M \right) + \left( {{N_1}M - 1} \right)$.
							Based on (\ref{TSAULAA}), we observe that the holes in ${\mathbb{SH}}_{\rm{l_{11}}}$ are filled by ${\mathbb{P}_{\rm{{t_1}}}} + {\mathbb{P}_{\rm{t_{1}}}} = \left\langle {1, 2{N_1} - 1} \right\rangle (M)$ and ${\mathbb{P}_{\rm{t_1}}} + {\mathbb{P}_{\rm{t_3}}}  = \left( {{N_1}M + 1} \right) + \left\{ {\left\langle {1:{N_1}} \right\rangle \left( M \right) + \left\langle {0:M/2 - 2} \right\rangle \left( 2 \right)} \right\}$ fill the holes ${\mathbb{SH}}_{\rm{l_{12}}}$.
							As the resulting co-array is combination of sum and difference sets, the remaining holes locations, ${\mathbb{SH}}_{\rm{l_{13}}}$, can be filled by difference set ${\mathbb{P}_{\rm{t_1}}} - {\mathbb{P}_{\rm{t_{2a}}}}$ = $\left\langle {0:{N_1}} \right\rangle \left( M \right) + \left( {{N_1}M - 1} \right).$
							
							Moreover, ${\mathbb{P}_{\rm{t_{2a}}}} + {\mathbb{P}_{\rm{t_{2a}}}}$ is continuous over $\left\langle {2{N_1}M - M + 4, 2{N_1}M + M - 4} \right\rangle$ except at holes specified by ${\mathbb{SH}_{\rm{L_2}}} = \left( {2{N_1}M - M + 5} \right) + \left\langle {0:M - 5} \right\rangle \left( 2 \right)$. From the definition of the sum set ${\mathbb{P}_{\rm{t_1}}} + {\mathbb{P}_{\rm{t_3}}}$, we notice that it produces virtual lags at these locations, except for the hole at $\left( {2{N_1}M - M + 5} \right)$, which is filled by ${\mathbb{P}_{\rm{t_1}}} - {\mathbb{P}_{\rm{t_{2b}}}}$. 
							Similarly, ${\mathbb{P}_{\rm{t_{2a}}}} + {\mathbb{P}_{\rm{t_3}}}$ is continuous in $ \left\langle 2{N_1}M + 3, 2{N_1}M + 2M - 5 \right\rangle$, except for holes at ${\mathbb{SH}}_{\rm{l_{23}}} = {(2{N_1}M + 2) + \left\langle {1, M - 4} \right\rangle (2)}.$
							These holes can be filled by  $({\mathbb{P}_{\rm{t_{2b}}}} + {\mathbb{P}_{\rm{t_{2b}}}}) \cup ({\mathbb{P}_{\rm{t_3}}} - {\mathbb{P}_{\rm{t_{2b}}}}),$ in which case, the holes located at ${(2{N_1}M + 2) + \left\langle {1, M/2 - 2} \right\rangle (2)}$ 
							are filled by ${\mathbb{P}_{\rm{t_{2b}}}} + {\mathbb{P}_{\rm{t_{2b}}}}$, while the the difference set  ${\mathbb{P}_{\rm{t_3}}} - {\mathbb{P}_{\rm{t_{2b}}}} = (2{N_1}M + M - 2 + \left\langle {1:M/2 - 1} \right\rangle \left( 2 \right))$ fill the remaining holes in the range ${(2{N_1}M + M - 2) + \left\langle {1, M/2 - 2} \right\rangle (2)}$. 
							Hence, the lags are continuous up to  $2{N_1}M + 2M - 5$. 
							
							Finally, the positive sum set  ${\mathbb{P}_{\rm{t_3}}} + {\mathbb{P}_{\rm{t_3}}} = {(2{N_1}M + M) + \left\langle {1, (M + 2)/2} \right\rangle}(2)$ is continuous in the range $\left\langle 2{N_1}M + M + 2, 2{N_1}M + 2M + 2 \right\rangle$, except for holes at ${\mathbb{SH}}_{\rm{l_{33}}} = {(2{N_1}M + M + 1) + \left\langle {1, M - 4} \right\rangle (2)}$. Among these, the holes at ${(2{N_1}M + M + 1) + \left\langle {1, M/2 - 3} \right\rangle (2)}$ overlap with the virtual lags produced by ${\mathbb{P}_{\rm{t_{2a}}}} + {\mathbb{P}_{\rm{t_3}}}$. As a result, the joint sum-difference co-array of the TSAULAs has no gaps (holes) in  $\pm {l_{t_1}}$, where $ {l_{t_1}} = 2{N_1}M + 2M - 4$. Hence, the total number of uDOFs is obtained by $4{N_1}M + 4M - 7$. 
							
							\section*{Appendix D}
							\textcolor{black}{Since \textbf{Lemma 4} covers the co-array properties of Co-TSAULAs, we use the definitions provided in equations (\ref{Co-TSAULAA})-(\ref{Co-TSAULAA3}) for this proof.}
							
							The Co-TSAULAs have the same structure as the TSAULAs, except one element from the first sparse ULA of TSAULAs is relocated and collinearly positioned to the third ULA at a distance of $d$. Therefore, the proof of difference and sum sets of Co-TSAULAs follow a similar line of reasoning and can be easily obtained by replacing the location sets of TSAULAs, ${\mathbb{P}_{\rm{t_1}}}, {\mathbb{P}_{\rm{t_2}}}, {\mathbb{P}_{\rm{t_3}}}$, in Appendix C with the locations sets for Co-TSAULAs defined in (\ref{Co-TSAULAA}), ${\mathbb{P}_{\rm{c_1}}}, {\mathbb{P}_{\rm{c_2}}}, {\mathbb{P}_{\rm{c_3}}}$. Such as, the sum set ${\mathbb{P}_{\rm{c_3}}} + {\mathbb{P}_{\rm{c_3}}} = {(2{N_1}M + M + 2) + \left\langle {M-4, 2M-6} \right\rangle} \cup {(2{N_1}M + M + 2) + \left\langle {0, M/2-3} \right\rangle}(2)$ is continuous in the range $\left\langle {2{N_1}M + M + 2, 2{N_1}M + 3M - 4} \right\rangle$
							except for holes at ${\mathbb{SH}}_{\rm{l_{33}}} = {(2{N_1}M + M + 3) + \left\langle {0, M/2-3} \right\rangle}(2)$. 
							Since $N_{1} = N - M$ for Co-TSAULAs, these holes can be  filled by ${\mathbb{P}_{\rm{c_2}}} + {\mathbb{P}_{\rm{c_3}}}$. Hence, there are no holes in the joint sum-difference co-array, containing continuous virtual lags in ${\left\langle {-2{N_1}M - 3M + 4}, {2{N_1}M + 3M - 4} \right\rangle}$. Therefore, the DOFs for Co-TSAULAs are obtained by $4{N_1}M + 6M - 7$.
							\bibliographystyle{IEEEtran}
							\bibliography{ReferencesCleaned}

						\end{document}